\documentclass[onecolumn,aps,nofootinbib,superscriptaddress,tightenlines,notitlepage,11pt,pra]{revtex4-1}

\usepackage{amsmath}
\usepackage{amsthm}
\usepackage{amsfonts}
\usepackage{color}

\usepackage[ breaklinks,
			 colorlinks = true,
             linkcolor = blue,
             urlcolor  = blue,
             citecolor = red,
             anchorcolor = blue,
]{hyperref}

\newtheorem{lemma}{Lemma}

\newtheorem{theorem}{Theorem}

\usepackage[capitalise]{cleveref}

\DeclareMathOperator{\tr}{tr}
\newcommand{\id}{1}

\newcommand{\eps}{\varepsilon}
\newcommand{\cS}{\mathcal{S}}

\newcommand{\cP}{\mathcal{P}}

\newcommand{\cE}{\mathcal{E}}

\newcommand{\suppress}[1]{}

\newtheorem{definition}{Definition}

\begin{document}

\title{\Large Partially smoothed information measures}

\author{Anurag Anshu}
\affiliation{Center for Quantum Technologies, National University of Singapore, Singapore}
\author{Mario Berta}
\affiliation{Department of Computing, Imperial College London, London, United Kingdom}
\author{Rahul Jain}
\affiliation{Center for Quantum Technologies, National University of Singapore and MajuLab, UMI 3654, Singapore}
\author{Marco Tomamichel}
\affiliation{Centre for Quantum Software and Information, University of Technology Sydney, Sydney, Australia}


\begin{abstract}
Smooth entropies are a tool for quantifying resource trade-offs in (quantum) information theory and cryptography. In typical bi- and multi-partite problems, however, some of the sub-systems are often left unchanged and this is not reflected by the standard smoothing of information measures over a ball of close states. We propose to smooth instead only over a ball of close states which also have some of the reduced states on the relevant sub-systems fixed. This partial smoothing of information measures naturally allows to give more refined characterizations of various information-theoretic problems in the one-shot setting. In particular, we immediately get asymptotic second-order characterizations for tasks such as privacy amplification against classical side information or classical state splitting. For quantum problems like state merging the general resource trade-off is tightly characterized by partially smoothed information measures as well.
\end{abstract}

\maketitle


\section{Introduction}

One-shot information theory concerns itself with finding tight bounds on the resource trade-offs for various operational problems in information theory and cryptography (see, e.g., \cite{mybook} for an introduction). Smooth entropies and smooth mutual informations have in many cases proven to be adequate information measures in this context. On the one hand, smooth min-entropy was first introduced in the context of quantum cryptography~\cite{renner05}. More precisely, the smooth conditional min-entropy was introduced to characterize the amount of uniform and independent randomness that can be extracted from a correlated random variable. On the other hand, the smooth max-information has been introduced to quantify the communication requirements in quantum extensions of Slepian-Wolf coding~\cite{bertachristandl11}. Since then smooth entropy measures of various kinds have been used to characterize a plethora of other tasks as well.

Our main contributions can be summarized as follows:
\begin{itemize}
  \item We introduce a notion of partially smoothed mutual max-information and conditional min-entropy and establish some of their basic mathematical properties.
  \item We show that these new definitions are equivalent to their fully smoothed counterparts, up to terms that vanish in the first-order i.i.d.\ asymptotics. Moreover, for the fully classical case this equivalence even holds for the asymptotic second-order i.i.d.\ asymptotics.
  \item We give several examples of operational problems where the new quantities naturally appear to give tighter bounds for the one-shot problem. In particular, for classical problems this leads to asymptotic second-order i.i.d.\ expansions.
\end{itemize}

The remainder of this paper is organized as follows. In~\cref{sec:def} we introduce our new definition of smooth mutual max-information and conditional min-entropy with restricted smoothing. In~\cref{sec:equi} we show that these definitions are, up to small correction terms, equivalent to the standard definitions found in the literature. We are able to show stronger equivalences for the special case of classical distributions. Finally, \cref{sec:op} discusses various operational interpretations of the new measures in detail, both in the classical and quantum context.


\section{Definitions and basic properties}\label{sec:def}

\subsection{Basic notation}

For any finite-dimensional inner product space describing a quantum system\,---\,let us fix it to $A$ for the sake of clarity\,---\, we define the set of positive semi-definite (psd) operators acting on $A$ as $\cP(A)$. We also define two subsets: the set of quantum states (i.e.\ psd operators with unit trace), denoted $\cS_{\circ}(A)$, and sub-normalized states (i.e.\ psd operators with trace not exceeding unity), denoted $\cS_{\bullet}(A)$. We describe joint quantum systems using the shorthand $AB = A \otimes B$.

The L\"owner partial order of operators in $\cP(A)$, denoted by $'\geq'$, is given by the relation $A \geq B$ if and only if $A - B$ is psd. Moreover, we say that $A$ dominates $B$, denoted $A \gg B$, if and only if the support of $B$ is contained in the support of $A$. 


\subsection{Smooth entropy measures}

Let $\rho$ and $\sigma$ be two psd operators. If $\rho \ll \sigma$ we define the max-divergence~\cite{jain02,datta08} as
\begin{align}
  D_{\max}(\rho\|\sigma) := \inf \big\{ \lambda : \rho \leq \exp(\lambda) \sigma \big\} \,,
\end{align}
and otherwise it is defined as $+\infty$. This quantity can be used to define various notions of mutual max-information and conditional min-entropy, respectively. We will concern ourselves with the following two definitions~\cite{renner05,bertachristandl11}.
For any bipartite state $\rho_{AB} \in \cS_{\circ}(AB)$, we have
\begin{align}
  I_{\max}(A ; B)_{\rho} &:= \inf_{\sigma_B \in \cS_{\bullet}(B)} D_{\max}(\rho_{AB} \| \rho_A \otimes \sigma_B) \,,\label{eq:Imax}\\
  H_{\min}(A | B)_{\rho} &:= - D_{\max}(\rho_{AB} \| \id_A \otimes \rho_B) \,.\label{eq:Hmin}
\end{align}
One immediate observation is that $\sigma_B \in \cS_\circ(B)$ would have worked equally well since the minimizer will always have maximal trace. Moreover, many variations of these definitions can be found in the literature. Most prominently, the min-entropy can be defined using a maximization over $\sigma_B \in \cS_{\bullet}(B)$ similar to the max-information, which yields a quantity with clear operational interpretation~\cite{renner05,koenig08} (see also~\cite{ciganovic13} about the max-information). The above choices are determined by the applications we discuss in Section~\ref{sec:op}.

Our goal is to define a smooth max-information and smooth min-entropy based on the above quantities, i.e.\ quantities for which $\rho_{AB}$ is replaced with a ball of states close to $\rho_{AB}$. In particular, we want the states in this ball to have the property that the $A$ subsystem is (essentially) left intact. To do this we will need to use a metric on (sub-normalized) quantum states, i.e.\ positive semi-definite operators with trace not exceeding 1. This metric, let us denote it by $\Delta(\cdot, \cdot)$, needs to have the following properties (we assume $\rho, \sigma$ and $\tau$ are sub-normalized quantum states).
\begin{enumerate}
  \item Positive definiteness: $\Delta(\rho, \sigma) \geq 0$ with equality only if $\rho = \sigma$.
  \item Triangle inequality: $\Delta(\rho, \sigma) \leq \Delta(\rho, \tau) + \Delta(\tau, \sigma)$.
  \item Strong monotonicity: $\Delta(\cE(\rho), \cE(\sigma)) \leq \Delta(\rho, \sigma)$ for any completely positive and trace non-increasing map $\cE$.
\end{enumerate}

We in particular will consider two metrics satisfying the above properties. The purified distance~\cite{tomamichel09} based on the generalized fidelity, 
\begin{align}
  P(\rho,\tau) := \sqrt{1 - F^2(\rho,\tau)} \quad \textrm{with} \quad
  F(\rho,\tau) := \tr\Big[\big|\sqrt{\rho}\sqrt{\tau} \big|\Big] + \sqrt{(1-\tr[\rho])(1-\tr[\tau])} \,,
\end{align}
and the generalized trace distance~\cite{mybook},
\begin{align}
  T(\rho,\tau) := \frac12 \tr\Big[\big| \rho - \tau \big|\Big] + \frac12 \big| \tr[\rho] - \tr[\tau] \big| \,.
\end{align}
We have the equivalence
\begin{align}\label{eq:pt-equivalence}
T(\rho,\tau)\leq P(\rho,\tau)\leq\sqrt{2T(\rho,\tau)}
\end{align}
from~\cite[Lem.~3.5]{mybook}. We will also use the abbreviations for the standard measures
\begin{align}
\bar{F}(\rho,\tau):=\tr\Big[\big|\sqrt{\rho}\sqrt{\tau} \big|\Big]\quad \textrm{and} \quad \|\rho-\tau\|_1:=\tr\Big[\big| \rho - \tau \big|\Big]\,.
\end{align}
In the following, let $\Delta$ be a metric satisfying the above Properties 1--3. Moreover, we call a tuple $(\eps, \Delta)$ with $\eps \geq 0$ \emph{valid} for a state $\rho$ if $\Delta(\rho, 0) > \eps$, where $0$ denotes the additive identity.\footnote{This is simply to ensure that the all zero matrix is not in an $\eps$-neighbourhood of $\rho$.} The following two definitions are rather standard (see, e.g., \cite{mybook} for an overview):

\begin{definition}\label{def:oldsmooth}
  Let $\rho_{AB} \in \cS_{\circ}(AB)$ with valid $(\eps, \Delta)$. The {$(\eps,\Delta)$-smooth max-information of $A$ and $B$} is defined as\footnote{The original definition of the smooth max-information in~\cite[Eq.~12]{bertachristandl11} was slightly different and based on $D_{\max}(\tilde{\rho}_{AB} \| \tilde{\rho}_A \otimes \sigma_B)$.}
  \begin{align}
    && I_{\max}^{\eps, \Delta}( A ; B)_{\rho}\ :=\ \inf\quad & D_{\max}(\tilde{\rho}_{AB} \|\rho_A \otimes \sigma_B) &&&&\\
    && \textnormal{s.t.}\quad & \tilde{\rho}_{AB} \in \cS_{\circ}(AB), \notag\\
    &&& \Delta(\tilde{\rho}_{AB},\rho_{AB}) \leq \eps, \notag\\
    &&&  \sigma_B \in \cS_{\circ}(B) \notag\,.
  \end{align}
  Moreover, the {$(\eps,\Delta)$-smooth conditional min-entropy of $A$ given $B$} is defined as
  \begin{align}
    && H_{\min}^{\eps, \Delta}( A | B)_{\rho}\ :=\ \sup\quad & - D_{\max}(\tilde{\rho}_{AB} \| \id_A \otimes \rho_B) &&&&\\
    && \textnormal{s.t.}\quad & \tilde{\rho}_{AB} \in \cS_{\bullet}(AB), \notag\\
    &&& \Delta(\tilde{\rho}_{AB}, \rho_{AB}) \leq \eps\,.
  \end{align}
\end{definition}

As for the non-smooth case there are other definitions in use that we will not discuss here specifically, e.g.\ in the definition of the smooth max-information one can fix $\sigma_B$ to be $\rho_B$ to arrive at a different quantity, and similarly the min-entropy can be further optimized over $\sigma_B \in \cS_{\bullet}(B)$. Note that for the smooth min-entropy it is necessary to smooth over sub-normalized states as otherwise the quantity will not be invariant under the application of local embedding maps (see also~\cite[Sec.~6.2.3]{mybook}). 

Given this, we propose the following definition for the smooth max-information.\footnote{A similar locally smoothed quantity has previously made an appearance in~\cite[Def.~1.15]{Anshu17-2} as a proof tool.}

\begin{definition}\label{def:smoothmaxinfo}
  Let $\rho_{AB} \in \cS_{\circ}(AB)$ with valid $(\eps, \Delta)$. The \emph{$(\eps,\Delta)$-smooth max-information with fixed $A$ of $\rho_{AB}$} is defined as
  \begin{align}
     && I_{\max}^{\eps, \Delta}( \dot{A} ; B)_{\rho}\ :=\ \inf\quad & D_{\max}(\tilde{\rho}_{AB} \| \rho_A \otimes \sigma_B) &&&&\\
     && \textnormal{s.t.}\quad & \tilde{\rho}_{AB} \in \cS_{\circ}(AB), \notag\\
     &&& \Delta(\tilde{\rho}_{AB},\rho_{AB}) \leq \eps, \notag\\
     &&& \tilde{\rho}_A = \rho_A, \notag\\
     &&&  \sigma_B \in \cS_{\circ}(B) \notag\,.
  \end{align}
\end{definition}

For the purpose of this paper we also suggest the following definition of smooth conditional min-entropy.

\begin{definition}\label{def:smoothminentropy}
  Let $\rho_{AB} \in \cS_{\circ}(AB)$ with valid $(\eps, \Delta)$. Then, the \emph{$(\eps,\Delta)$-smooth min-entropy with fixed $B$ of $\rho_{AB}$} is defined as
  \begin{align}
    && H_{\min}^{\eps, \Delta}(A | \dot{B})_{\rho}\ :=\ \sup\quad & -D_{\max}(\tilde{\rho}_{AB} \| \id_A \otimes \rho_B) &&&&\\
    && \textnormal{s.t.}\quad & \tilde{\rho}_{AB} \in \cS_{\bullet}(AB), \notag\\
    &&& \Delta(\tilde{\rho}_{AB},\rho_{AB}) \leq \eps, \notag\\
    &&& \tilde{\rho}_B \leq \rho_B, \notag\,. 
  \end{align}
\end{definition}

Note that smooth versions of all conditional R\'enyi entropies (see, e.g.,~\cite{tomamichel13}) can be defined analogously. However, we will not explore these definitions further here.

If the input states are classical in a fixed basis all the definitions apply for this case as well. It is then immediate to see that the respective optimizations over $\tilde{\rho}_{AB}$ and $\sigma_B$ can without loss of generality be restricted to be diagonal in this fixed basis as well.\footnote{To see this, for example for the smooth max-information, note that if this were not so then the full dephasing map (in the classical basis) could be applied to both sides of the operator inequality
\begin{align}
    \tilde{\rho}_{AB} \leq \rho_A \otimes \sigma_B \,,
\end{align}
yielding a new feasible solution since the distance between $\rho_{AB}$ and $\tilde{\rho}_{AB}$ is also reduced when the dephasing map is applied due to Lem.~3.}


\subsection{Basic properties}

We will now discuss some basic properties of the quantities introduced above. In the following lemmas we assume that $\Delta$ satisfies Properties 1--3. Let us explore the above two definitions. The first property is an immediate consequence of the positive definiteness of $\Delta$.

\begin{lemma}
Let $\rho_{AB} \in \cS_{\circ}(AB)$. Then, we have
\begin{align}
I_{\max}^{0, \Delta}(\dot{A} ; B)_{\rho} = I_{\max}(A ; B)_{\rho}
\qquad \textrm{and} \qquad
H_{\min}^{0, \Delta}(A|\dot{B})_{\rho} = H_{\min}(A|B)_{\rho}\,.
\end{align}
\end{lemma}

The second lemma, on the other hand, is a direct consequence of the triangle inequality of $\Delta$.

\begin{lemma}
  Let $\rho_{AB}, \tilde{\rho}_{AB} \in \cS_{\circ}(AB)$ with $\Delta(\rho_{AB}, \tilde{\rho}_{AB}) \leq \eta$, and $(\eps\!+\!\eta, \Delta)$ valid for $\rho_{AB}$. Then, we have
  \begin{align}
    I_{\max}^{\eps+\eta, \Delta}(\dot{A} ; B)_{\rho} \leq I_{\max}^{\eps, \Delta}(\dot{A} ; B)_{\tilde{\rho}} 
    \qquad \textrm{and} \qquad
    H_{\min}^{\eps+\eta, \Delta}(A ; \dot{B})_{\rho} \geq H_{\min}^{\eps, \Delta}(A ; \dot{B})_{\tilde{\rho}} 
    \,.
  \end{align}  
\end{lemma}

We continue with the following observation that follows immediately from well-known properties of the max-divergence, namely that the smooth max-information is non-increasing under local completely positive trace-preserving (cptp) operations and the smooth min-entropy is non-decreasing under local operations.

\begin{lemma}\label{lm:cptp}
  Let $\rho_{AB} \in \cS_{\circ}(AB)$ with valid $(\eps, \Delta)$. For any two completely positive trace preserving maps $\mathcal{E}: \cP(A) \to \cP(A')$ and $\mathcal{F}: \cP(B) \to \cP(B')$, we have
  \begin{align}
    I_{\max}^{\eps}( \dot{A} ; B)_{\rho} \geq I_{\max}^{\eps}( \dot{A}' ; B')_{\tau} 
  \end{align}
  where $\tau_{A'B'} = (\mathcal{E} \otimes \mathcal{F})(\rho_{AB})$. Furthermore, if $\mathcal{E}$ is also sub-unital (i.e.\ it satisfies $\mathcal{E}(\id_A) \leq \id_{A'}$), then
  \begin{align}
    H_{\min}^{\eps}(A ; \dot{B})_{\rho} \leq H_{\min}^{\eps}(A' ; \dot{B}')_{\tau} \,.
  \end{align}
\end{lemma}

Clearly every operational definition should be invariant under isometries as embeddings are essentially just a choice of modeling and should not effect operational quantities.

\begin{lemma}\label{lm:iso}
  Let $\rho_{AB} \in \cS_{\circ}(AB)$ with valid $(\eps, \Delta)$. For any two isometries $U: A \to A'$, $V: B \to B'$, it holds that
  \begin{align}
    I_{\max}^{\eps,\Delta}(\dot{A} ; B)_{\rho} = I_{\max}^{\eps,\Delta}(\dot{A}' ; B')_{\rho} 
    \qquad \textrm{and} \qquad  
    H_{\min}^{\eps,\Delta}(A ; \dot{B})_{\rho} = H_{\min}^{\eps,\Delta}(A' ; \dot{B}')_{\rho} \,,
  \end{align} 
  where $\rho_{A'B'} = (U \otimes V) \rho_{AB} (U \otimes V)^{\dag}$. 
\end{lemma}

\begin{proof}
  The result for the smooth max-information follows immediately from Lem.~\ref{lm:cptp}. To see this, note that the map $(U \otimes V) \cdot (U \otimes V)^{\dag}$ is cptp and we can define cptp inverse maps of the form 
\begin{align}
  \rho_A \mapsto U^{\dag} \rho_A U + \big(1 - \tr U^{\dag} \rho_A U \big) \tau_A, \qquad
  \rho_B \mapsto V^{\dag} \rho_B V + \big(1 - \tr V^{\dag} \rho_B V \big) \tau_B, \label{eq:inverseiso}
\end{align}
where $\tau_{A} \in \cS_{\circ}(A)$ and $\tau_{B} \in \cS_{\circ}(B)$ are arbitrary states.

  We now present the proof for the smooth min-entropy by proving inequalities in both directions. The direction `$\geq$' is guaranteed by Lem.~\ref{lm:cptp} with the same argument as above. However, the map in~\eqref{eq:inverseiso} is not sub-unital so we cannot
  employ the data-processing inequality to show `$\leq$'. Instead, consider the state $\tilde{\rho}_{A'B'}$ optimal for $H_{\min}^{\eps,\Delta}(A' | \dot{B}')_{\rho}$. We define
  \begin{align}
    \tilde{\rho}_{AB} := (U \otimes V)^{\dag} \tilde{\rho}_{A'B'} (U \otimes V)
  \end{align}
  and observe that $\rho_B= V^{\dag} \rho_{B'} V$.
  Note that the maps $U^{\dag} (\cdot) U$ and $V^{\dag} (\cdot) V$ are in general not trace-preserving as weight outside the range of $U$ and $V$ is discarded when we invert the isometries. However, the resulting state $\tilde{\rho}_{AB}$ and $\rho_B$ are feasible for the optimization in $H_{\min}^{\eps,\Delta}(A | \dot{B})_{\rho}$ since the following holds:

  \begin{itemize}
  \item We have $\tilde{\rho}_{AB} \in \cS_{\bullet}(AB)$ and $\rho_B \in \cS_{\bullet}(B)$\,---\,that for an isometry $V$ the operator $VV^{\dagger}$ is a projection onto the image of $V$;
  \item It holds that $\Delta(\tilde{\rho}_{AB}, \rho_{AB}) \leq \Delta(\tilde{\rho}_{A'B'}, \rho_{A'B'}) \leq \eps$ due to the fact that the metric is monotone under trace non-increasing completely positive maps;
  \item We have $\tilde{\rho}_B = V^{\dag} \tr_A \left( (U^{\dag}\otimes \id_B) \tilde{\rho}_{A'B'}  (U\otimes \id_B) \right) V \leq V^{\dag} \tilde{\rho}_{B'} V \leq V^{\dag} \rho_{B'} V = \rho_B$, where the first inequality is due to Lem.~\ref{lm:someineq} in App.~\ref{app:lemmas}.
  \end{itemize}
  Finally, the desired inequality follows from the implication:
    \begin{align}
    \tilde{\rho}_{A'B'} \leq \exp(\lambda) \id_{A'} \otimes \rho_{B'} \quad \implies \quad \tilde{\rho}_{AB} \leq \exp(\lambda) U^{\dag} \id_{A'} U \otimes \rho_B  \,,
  \end{align}
  that we yield from applying the completely positive map $(U \otimes V)^{\dag} \cdot (U \otimes V)$ on both sides of the operator inequality. Finally, note that $U^{\dag} \id_{A'} U = \id_A$.
\end{proof}

One could hope to replace $\tilde{\rho}_B \leq \rho_B$ in Def.~\ref{def:smoothmaxinfo} by an equality, thus forcing the state $\tilde{\rho}_{AB}$ to have the same trace as $\rho_{AB}$. However, for such a definition one would then need to show a property analogous to the above invariance under isometries, which does not obviously hold. The following argument gives also an indication that sub-normalized states are desirable in this context, although it does not conclusively show that they are necessary for our definition.

For the (unconditional) min-entropy, invariance under isometries can only hold if we allow sub-normalized states. To see this, consider the min-entropy of the state $\rho = \id / d$, which is maximal for normalized states of dimension $d$ and thus cannot be increased by smoothing over this set. However, if embedded into a larger space smoothing will yield a larger min-entropy. Allowing sub-normalized states introduces an alternative to moving weight out of the support of $\rho$ and it turns out that this is exactly what is needed to ensure the quantity is invariant under isometries.


\section{Relation to other Entropy Measures}\label{sec:equi}

\subsection{Classical Setting}\label{sec:equi/classical}

Since the (generalized) trace distance is directly connected to error probabilities it is often natural to stick to this distance measure for classical problems. We will do so in this section. We will also continue using the notations $\cP, \cS_{\circ}, \cS_{\bullet}$, although now we restrict to \textit{diagonal} matrices in some basis, interpreted as (potentially sub-normalized) probability distributions. In order to establish an asymptotic equipartition property for our locally smoothed information measures we relate them to other well-studied entropic quantities such as information spectrum divergences~\cite{han02}. Note that standard asymptotic equipartition proofs for mutual information and conditional entropy do not leave any of the marginals unchanged.

\begin{definition}
For $P_X,Q_X\in\cP(X)$ and $\eps\in[0,1]$, the max-information spectrum divergence is defined as
\begin{align}
D_s^{\eps}(P_X\|Q_X): = \inf\left\{a: \Pr_{x \leftarrow p_X}\left\{\frac{P_X(x)}{Q_{X}(x)}> 2^a\right\}<\eps\right\}\,.
\end{align}
\end{definition}

Importantly, the max-information spectrum divergence has the following asymptotic second-order expansion in for independent and identically distributed (i.i.d.) case~\cite{strassen62}
\begin{align}\label{eq:spectrum-divergence_second-order}
\frac{1}{n}D_s^{\eps}(P_X^{\times n}\|Q_X^{\times n})=D(P_X\|Q_X)+\sqrt{\frac{V(P_X\|Q_X)}{n}}\cdot\Phi^{-1}(\eps)+O\left(\frac{\log n}{n}\right)
\end{align}
with the Kullback-Leibler divergence $D(P_X\|Q_X):=\sum_xP_X(x)\log\left(\frac{P_X(x)}{Q_X(x)}\right)$, the information variance $V(P_X\|Q_X):=\mathbb{E}\left[(\log P_X-\log Q_X-D(P_X\|Q_X))^2\right]$, and the cumulative standard Gaussian distribution
\begin{align}
\Phi(x):=\int_{-\infty}^x\frac{1}{2\pi}\exp(x^2/2)\,\mathrm{d}x\,.
\end{align}
We then define the information spectrum max-information and conditional min-entropy as
\begin{align}
I_s^{\eps}(X;Y)_P:= D_s^{\eps}(P_{XY} \| P_X\times P_Y)
\qquad \textrm{and} \qquad  
H_s^{\eps}(X|Y)_P:= -D_s^{\eps}(P_{XY} \| 1_X\times P_Y)\,,
\end{align}
respectively. This leads to the following equivalence result.

\begin{theorem}\label{thm:classical_equivalence}
Let $P_{XY}\in\cS_{\circ}(XY)$ and $0<\eps+\delta\leq1$. Then, we have
\begin{align}\label{eq:IsImaxHsHmin}
&I_s^{\frac{\eps}{1-\delta}+\delta}(X;Y)_P - 2\log\frac{1}{\delta}\leq I_{\max}^{\eps, T}( \dot{X} ; Y)_P\leq I_s^{\eps}(X;Y)_P+1\\
&H_s^{\frac{\eps}{1-\delta}}(X|Y)_P + \log\frac{1}{\delta}\geq H_{\min}^{\eps, T}(X| \dot{Y})_P\geq H_s^{\eps}(X|Y)_P - 1\,.
\end{align}
In particular, this implies the asymptotic second-order expansions
\begin{align}
\frac{1}{n}I_{\max}^{\eps,T}(\dot{X};Y)_{P^{\times n}}&=I(X;Y)_P+\sqrt{\frac{V(X;Y)_P}{n}}\cdot\Phi^{-1}(\eps)+O\left(\frac{\log n}{n}\right)\\
\frac{1}{n}H_{\min}^{\eps, T}(X| \dot{Y})_{P^{\times n}}&=H(X|Y)_P+\sqrt{\frac{V(X|Y)_P}{n}}\cdot\Phi^{-1}(\eps)+O\left(\frac{\log n}{n}\right)
\end{align}
with the mutual information variance $V(X;Y)_P:=V(P_{XY}\|P_X\times P_Y)$ and the conditional information variance $V(X|Y)_P:=V(P_{XY}\|1_x\times P_Y)$.
\end{theorem}

For the proof of Thm.~\ref{thm:classical_equivalence} we first need to introduce some additional quantities and lemmas. Recall that for the classical special case we have
\begin{align}
I_{\max}^{\eps, T}(\dot{X}; Y)_{P} =  \inf_{P'_{XY}\in \cS_{\circ}(XY), Q_Y\in \cS_{\circ}(Y): P'_X = P_X, T(P'_{XY}, P_{XY})\leq \eps} D_{\max}(P'_{XY}\| P_X\times Q_Y)\,.
\end{align}
Now, with $Q_Y\in\cS_{\circ}(Y)$ we define the following intermediate quantities
\begin{align}
I_s^{\eps}(X;Y)_{P|Q}&:= D_s^{\eps}(P_{XY} \| P_X\times Q_Y)\\
I_{\max}^{\eps, T}(\dot{X}; Y)_{P|Q}&:= \inf_{P'_{XY} \in \cS_{\circ}(XY): P'_X = P_X, T(P'_{XY}, P_{XY})\leq \eps} D_{\max}(P'_{XY}\| P_X\times Q_Y)\,,
\end{align}
where their utility is that they roughly capture both the smooth max-information and the smooth min-entropy (as we will see). To continue, note that
\begin{align}\label{eq:reducetoImax}
&I_s^{\eps}(X;Y)_P = I_s^{\eps}(X;Y)_{P|P}\\
&H_s^{\eps}(X|Y)_P = \log|X| -  I_s^{\eps}(Y; X)_{P|U}\quad\text{with $U_X$ the uniform distribution}\\
&I_{\max}^{\eps, T}(\dot{X}; Y)_{P} = \inf_{Q_Y}I_{\max}^{\eps, T}(\dot{X}; Y)_{P|Q}\\
&H_{\min}^{\eps, T}(X| \dot{Y})_P \geq \log|X| - I_{\max}^{\eps, T}(\dot{Y}; X)_{P|U}\quad\text{with $U_X$ the uniform distribution.}\label{eq:reducetoImax-extra}
\end{align}
Here, we observe that $H_{\min}^{\eps,T}(X| \dot{Y})_P$ is only lower bounded since it involves a supremum over sub-normalized distributions. The last ingredient is the following lemma, which says that $I_s^{\eps}(X;Y)_{P|Q}$ cannot be too small in comparison to $I_s^{\eps}(X;Y)_{P}$.

\begin{lemma}
Let $P_{XY}\in \cS_{\circ}(XY)$, $Q_Y\in \cS_{\circ}(Y)$, and $0 < \eps+\delta <1$. Then, we have
\begin{align}\label{eq:I_sQsmall}
I_s^{\eps}(X;Y)_{P|Q} \geq I_s^{\eps + \delta}(X;Y)_{P} - \log\frac{1}{\delta}\,.
\end{align}
\end{lemma}

\begin{proof}
Let $c:= I_s^{\eps}(X;Y)_{P|Q}$, $\mathrm{Bad}_1$ be the set of all $(x,y)$ for which $P_{XY}(x,y) \geq 2^cP_X(x)Q_Y(y)$, and $\mathrm{Bad}_2$ be the set of all $(x,y)$ for which $Q_Y(y) > \frac{1}{\delta}P_Y(y)$. Now, observe that
\begin{align}
\Pr_{(x,y) \leftarrow P_{XY}}\left\{\mathrm{Bad}_2\right\} =\Pr_{y \leftarrow P_Y}\left\{Q_Y(y) > \frac{1}{\delta}P_Y(y)\right\} = \sum_{y:Q_Y(y) > \frac{1}{\delta}P_Y(y)}P_Y(y) \leq \delta\sum_yQ_Y(y) \leq \delta\,.
\end{align}
For all $(x,y) \notin \mathrm{Bad}_1\cup \mathrm{Bad}_2$, we have
\begin{align}
P_{XY}(x,y) \leq \frac{2^c}{\delta}P_X(x)P_Y(y)\,.
\end{align}
Furthermore, we have
\begin{align}
\Pr_{(x,y)\leftarrow P_{XY}}\left\{\mathrm{Bad}_1\cup \mathrm{Bad}_1\right\} \leq \Pr_{(x,y)\leftarrow P_{XY}}\left\{\mathrm{Bad}_1 \right\} + \Pr_{(x,y) \leftarrow P_{XY}}\left\{\mathrm{Bad}_2\right\} \leq \eps+\delta\,,
\end{align}
which proves the claim.
\end{proof}

The proof of the equivalence result in Thm.~\ref{thm:classical_equivalence} is then as follows.

\begin{proof}[Proof of Thm.~\ref{thm:classical_equivalence}]
We first show that for every $Q_Y\in \cS_{\circ}(Y)$,
\begin{align}\label{eq:arbitQimax}
I_s^{\frac{\eps}{1-\delta}}(X;Y)_{P|Q} - \log\frac{1}{\delta}\leq I_{\max}^{\eps, T}( \dot{X} ; Y)_{P|Q}\leq I_s^{\eps}(X;Y)_{P|Q}+1\,.
\end{align}
\begin{itemize}
\item For the lhs the argument is similar to that given in \cite[Thm.~10]{Anshu17}. Let $d:= I_{\max}^{\eps, T}(\dot{X};Y)_{P|Q}$ and $P'_{XY}$ be the probability distribution achieving the infimum in the definition. Define the set
\begin{align}
\mathcal{A}:=\left\{(x,y): \frac{P_{XY}(x,y)}{P'_{XY}(x,y)} \geq \frac{1}{\delta}\right\}\,.
\end{align}
Using Lem.~\ref{lem:tracedistprob} (see Appendix \ref{app:lemmas}), we find
\begin{align}
\eps \geq T(P_{XY}, P'_{XY}) \geq P_{XY}(\mathcal{A}) - P'_{XY}(\mathcal{A}) \geq P_{XY}(\mathcal{A}) - \delta P_{XY}(\mathcal{A})= (1-\delta)P_{XY}(\mathcal{A})\,.
\end{align}
Thus, we get that $P_{XY}(\mathcal{A}) \leq \frac{\eps}{1-\delta}$. Moreover, for every $(x,y)\in \mathcal{A}^{\mathrm{c}}$, we have
\begin{align}
P_{XY}(x,y)\leq \frac{1}{\delta}P'_{XY}(x,y) \leq \frac{2^d}{\delta} P_X(x)Q_Y(y)\,.
\end{align}
Hence, we conclude that $I_s^{\frac{\eps}{1-\delta}}(X:Y)_{P|Q} \leq d + \log\frac{1}{\delta}$.
\item For the rhs let $c:= I_s^{\eps}(X;Y)_{P|Q}$. It holds that
\begin{align}
\Pr_{(x,y) \leftarrow P_{XY}}\left\{\frac{P_{Y\mid X=x}(y)}{Q_Y(y)}> 2^c\right\} \leq \eps\,.
\end{align}
Now, let $\eps_x := \Pr_{y \leftarrow P_{Y\mid X=x}}\left\{\frac{P_{Y\mid X=x}(y)}{Q_Y(y)}> 2^c\right\}$ and $\mathrm{Good}_x$ be the set of all $y$ that satisfy $\frac{P_{Y\mid X=x}(y)}{Q_Y(y)}\leq 2^c$. Then, we have $\eps = \sum_x P_X(x) \eps_x$. Define random variable $P'_Y$ jointly correlated with $P_X$ as
\begin{align}
P'_{Y\mid X=x}(y) = P_{Y\mid X=x}(y)\cdot \id(y\in \mathrm{Good}_x) + \eps_x Q_Y(y), \quad P'_{XY}(x,y)= P_X(x) P'_{Y\mid X=x}(y)\,.
\end{align}
We get $P'_{Y\mid X=x}(y) \leq (2^c + \eps_x) Q_Y(y) \leq (2^c+1) Q_Y(y)$, which implies $D_{\max}(P'_{XY}\| P_X\times Q_Y) \leq \log(2^c+1) \leq c+1$. Moreover, we calculate
\begin{eqnarray}
T(P_{XY}, P'_{XY}) &=& \frac{1}{2}\|P_{XY}-P'_{XY}\|_1 + \frac{1}{2}|P_{XY}(\id_{XY}) - P'_{XY}(\id_{XY})|\\
&=& \sum_x P_X(x)\frac{1}{2}\|P'_{Y\mid X=x} - P_{Y\mid X=x}\|_1 + 0 \\ &\leq& \sum_x P_X(x)\frac{1}{2}\left(\eps_x\|Q_Y(y)\|_1 + \|P_{Y\mid X=x}-P_{Y\mid X=x}\cdot\id(\mathrm{Good}_x)\|_1 \right) \\ &=& \sum_x P_X(x)\eps_x = \eps\,.
\end{eqnarray}
Since $P'_X = P_X$, we conclude that $I_{\max}^{\eps, T}(\dot{X};Y)_{P|Q}\leq I_s^{\eps}(X;Y)_{P|Q}+1$.
\end{itemize}
Eq.~\eqref{eq:IsImaxHsHmin} is now proved as follows:
\begin{itemize}
\item For the lhs let $Q_Y$ be the probability distribution achieving the infimum in $I_{\max}^{\eps, T}( \dot{X} ; Y)_P$. From Eq.~\eqref{eq:arbitQimax} and~\eqref{eq:I_sQsmall}, we obtain
\begin{align}
I_s^{\frac{\eps}{1-\delta}+\delta}(X;Y)_P - 2\log\frac{1}{\delta}\leq I_s^{\frac{\eps}{1-\delta}}(X;Y)_{P|Q} - \log\frac{1}{\delta} \leq I_{\max}^{\eps, T}( \dot{X} ; Y)_{P|Q}  = I_{\max}^{\eps, T}( \dot{X} ; Y)_P\,.
\end{align}
\item For the rhs we use Eq.~\eqref{eq:reducetoImax} and~\eqref{eq:arbitQimax} to conclude
\begin{align}
I_{\max}^{\eps, T}( \dot{X} ; Y)_P = \inf_{Q_Y}I_{\max}^{\eps, T}( \dot{X} ; Y)_{P|Q}\leq \inf_{Q_Y} I_s^{\eps}(X;Y)_{P|Q}+1 \leq  I_s^{\eps}(X;Y)_{P|P}+1 = I_s^{\eps}(X;Y)_{P}+1\,.
\end{align}
\end{itemize}
Now, we proceed to the min-entropy inequalities.  
\begin{itemize}
\item The first inequality is equivalent to first part of Eq.~\eqref{eq:arbitQimax}. Let $d:= H_{\min}^{\eps, T}(X|\dot{Y})_P$ and $P'_{XY}$ be the (sub normalized) random variables achieving the infimum in the definition. Let $\mathcal{A}:=\{(x,y): \frac{P_{XY}(x,y)}{P'_{XY}(x,y)} \geq \frac{1}{\delta}\}$. Then, using Lem.~\ref{lem:tracedistprob} we have
\begin{align}
\eps \geq T(P_{XY},P'_{XY}) \geq P_{XY}(\mathcal{A}) - P'_{XY}(\mathcal{A}) \geq P_{XY}(\mathcal{A}) - \delta P_{XY}(\mathcal{A})= (1-\delta)P_{XY}(\mathcal{A})\,.
\end{align}
Thus, we get $P_{XY}(\mathcal{A}) \leq \frac{\eps}{1-\delta}$. Moreover, for every $(x,y)\in \mathcal{A}^{\mathrm{c}}$, we have
\begin{align}
P_{XY}(x,y)\leq \frac{1}{\delta}P'_{XY}(x,y) \leq \frac{2^{-d}}{\delta} 1_X(x)P_Y(y)\,.
\end{align}
Thus, we find $H_s^{\frac{\eps}{1-\delta}}(X|Y)_P \geq d - \log\frac{1}{\delta}$.
\item For the second inequality let $c:= H^{\eps}_s(X|Y)_P$.  Then, we have
\begin{align}
\Pr_{(x,y) \leftarrow P_{XY}}\left\{\frac{P_{XY}(x,y)}{P_Y(y)}> 2^{-c}\right\} \leq \eps\,.
\end{align}
Let $U_X$ be uniform distribution over $X$ and let $c':= \log|X| - c$. Then, we find
\begin{align}
\Pr_{(x,y) \leftarrow P_{XY}}\left\{\frac{P_{X\mid Y=y}(x)}{U_X(x)}> 2^{c'}\right\} \leq \eps\,.
\end{align}
which implies $c' \geq I_{s}^{\eps}(Y; X)_{P|U}$. From Eq.~\eqref{eq:arbitQimax} and~\eqref{eq:reducetoImax-extra}, we find that
\begin{align}
c'\geq I_{s}^{\eps}(Y; X)_{P|U} \geq I_{\max}^{\eps, T}(\dot{Y}; X)_{P|U} -1 \geq \log|X| - H_{\min}^{\eps, T}(X| \dot{Y})_P -1\,,
\end{align}
and substituting the value of $c'$, the proof concludes.
\end{itemize}
\end{proof}

Instead of smoothing over nearby states that have one of the reduced states (essentially) left intact we could alternatively even smooth over nearby states that have both reduced states (essentially) left intact. This would follow the intuition to smooth the correlations between the systems while leaving the individual systems unchanged and naturally extends to multi-partite scenarios. By iteratively applying the methods from the proof of Thm.~\ref{thm:classical_equivalence} we then find similar equivalence statements with the same max-information spectrum divergence based measures. This again leads to an asymptotic equipartition property, however, the expansion only becomes asymptotically tight in first-order (and not in second-order). We note that for many operational problems it seems more adapted to only fix one of the reduced states (see \cref{sec:op}).


\subsection{Quantum setting}\label{sec:equi/quantum}

Uhlmann's theorem~\cite{uhlmann85} indicates that it is natural to work with fidelity based distance measures for fully quantum problems. As such, we will use the purified distance in this section. Now, the equivalence proof from \cref{sec:equi/classical} crucially uses the idea of conditioning on the classical side information and hence we cannot give a direct quantum analogue. Instead we find the following equivalence result with the standard smooth max-mutual information (based on a different proof technique).

\begin{theorem}\label{lem:quantum_eq}
Let $\rho_{AB}\in\cS_{\circ}(AB)$ and $0\leq2\eps+\delta\leq1$ with $\delta>0$. Then, we have
\begin{align}\label{eq:max_trace}
I^{2\eps+\delta, P}_{\max}( \dot{A} ; B)_\rho\leq I_{\max}^{\eps,P}(A;B)_\rho+\log\frac{8+\delta^2}{\delta^2}\,,
\end{align}
and by definition we also have the opposite inequality $I_{\max}^{\eps,P }(\dot{A} ; B)_\rho\geq I_{\max}^{\eps,P }(A;B)_\rho$.
\end{theorem}


\begin{proof}
Let $\tilde{\rho}_{AB}$ and $\sigma_B$ be the optimizers on the right-hand side of Eq.~\eqref{eq:max_trace}. Moreover, for some $\gamma>0$ let
\begin{align}
P_A^\gamma:=\left\{\frac{1}{\gamma}\tilde{\rho}_A-\rho_A\right\}_+\quad\text{and}\quad\bar{\rho}_{AB}:=P_A^\gamma\tilde{\rho}_{AB}P_A^\gamma\,,
\end{align}
where $\{X\}_+$ denotes the projector onto the positive part of a Hermitian operator $X$. Let $V_A$ be the unitary from the polar decomposition of $\rho_A^{\frac{1}{2}}\bar{\rho}_A^{\frac{1}{2}}$ such that
\begin{align}
F(\rho_A,\bar{\rho}_A)=\mathrm{Tr}\Bigg[\left|\rho_A^{\frac{1}{2}}\bar{\rho}_A^{\frac{1}{2}}\right|\Bigg]=\mathrm{Tr}\left[\rho_A^{\frac{1}{2}}\bar{\rho}_A^{\frac{1}{2}}V_A\right]\,.
\end{align}
For $\gamma=\frac{\delta^2}{8}$ define the bipartite quantum state
\begin{align}
\hat{\rho}_{AB}:=\underbrace{\rho^{\frac{1}{2}}_AV_A\bar{\rho}_A^{-\frac{1}{2}}\bar{\rho}_{AB}\bar{\rho}_A^{-\frac{1}{2}}V_A^\dagger\rho^{\frac{1}{2}}_A}_{=:\tau_{AB}}+\underbrace{\left(\rho^{\frac{1}{2}}_A(1_A-V_AP^\gamma_AV_A^\dagger)\rho^{\frac{1}{2}}_A\right)\otimes\sigma_B}_{=:\sigma_{AB}}\,,
\end{align}
which by inspection has $\hat{\rho}_A=\rho_A$. We calculate
\begin{align}
\hat{\rho}_{AB}\leq\;&\left\|(\rho_A\otimes\sigma_B)^{-\frac{1}{2}}\tilde{\rho}_{AB}(\rho_A\otimes\sigma_B)^{-\frac{1}{2}}\right\|_{\infty}\cdot\left(\rho^{\frac{1}{2}}_AV_A\bar{\rho}_A^{-\frac{1}{2}}P_A^\gamma\rho_AP_A^\gamma\bar{\rho}_A^{-\frac{1}{2}}V_A^\dagger\rho^{\frac{1}{2}}_A\right)\otimes\sigma_B\notag\\
&+\left(\rho_A^{\frac{1}{2}}(1_A-V_AP^\gamma_AV_A^\dagger)\rho_A^{\frac{1}{2}}\right)\otimes\sigma_B\,,
\end{align}
and by the definition of $P_A^\gamma$ we have $P_A^\gamma\rho_AP_A^\gamma\leq\frac{8}{\delta^2}\cdot\bar{\rho}_A$ as well as $1_A-P^\gamma_A\leq1_A$ leading to
\begin{align}\label{eq:proof1}
\hat{\rho}_{AB}\leq\left(\frac{8}{\delta^2}\cdot\left\|(\rho_A\otimes\sigma_B)^{-\frac{1}{2}}\tilde{\rho}_{AB}(\rho_A\otimes\sigma_B)^{-\frac{1}{2}}\right\|_{\infty}+1\right)\cdot \rho_A\otimes\sigma_B\,.
\end{align}
Using that $\left\|(\rho_A\otimes\sigma_B)^{-\frac{1}{2}}\tilde{\rho}_{AB}(\rho_A\otimes\sigma_B)^{-\frac{1}{2}}\right\|_\infty=2^{D_{\max}(\tilde{\rho}_{AB}\|\rho_A\otimes\sigma_B)}\geq1$~\cite[Lem.~6]{datta08} we get
\begin{align}\label{eq:proof2}
\frac{8}{\delta^2}\cdot\left\|(\rho_A\otimes\sigma_B)^{-\frac{1}{2}}\tilde{\rho}_{AB}(\rho_A\otimes\sigma_B)^{-\frac{1}{2}}\right\|_{\infty}+1\leq\frac{8+\delta^2}{\delta^2}\left\|(\rho_A\otimes\sigma_B)^{-\frac{1}{2}}\tilde{\rho}_{AB}(\rho_A\otimes\sigma_B)^{-\frac{1}{2}}\right\|_{\infty}\,.
\end{align}
Hence, the claim follows as soon as we establish that $\hat{\rho}_{AB}$ is close enough to $\rho_{AB}$ in purified distance. Now, notice that $\mathrm{Tr}[\sigma_{AB}]=1-\mathrm{Tr}[\tau_{AB}]$ and hence $\bar{\tau}_{AB}:=\frac{\tau_{AB}}{\mathrm{Tr}[\tau_{AB}]}$ and $\bar{\sigma}_{AB}:=\frac{\sigma_{AB}}{\mathrm{Tr}[1-\tau_{AB}]}$ are normalized. We can then write
\begin{align}
\hat{\rho}_{AB}=\mathrm{Tr}[\tau_{AB}]\cdot\bar{\tau}_{AB}+\big(1-\mathrm{Tr}[\tau_{AB}]\big)\cdot\bar{\sigma}_{AB}\,.
\end{align}
Since the fidelity $\bar{F}^2(\rho,\sigma)$ is concave in each argument (this follows from the operator concavity of the logarithm) we can estimate
\begin{align}
\bar{F}^2\left(\hat{\rho}_{AB},\bar{\rho}_{AB}\right)&\geq\mathrm{Tr}[\tau_{AB}]\cdot F^2\left(\bar{\tau}_{AB},\bar{\rho}_{AB}\right)+\big(1-\mathrm{Tr}[\tau_{AB}]\big)\cdot\bar{F}^2\left(\bar{\sigma}_{AB},\bar{\rho}_{AB}\right)\\
&\geq\mathrm{Tr}[\tau_{AB}]\cdot\bar{F}^2\left(\bar{\tau}_{AB},\bar{\rho}_{AB}\right)\\
&=\bar{F}^2\left(\tau_{AB},\bar{\rho}_{AB}\right)\,.\label{eq:previous}
\end{align}
By the triangle inequality for the purified distance we get for the quantity of interest
\begin{align}\label{eq:interest}
P\left(\hat{\rho}_{AB},\rho_{AB}\right)\leq P\left(\hat{\rho}_{AB},\bar{\rho}_{AB}\right)+P\left(\rho_{AB},\bar{\rho}_{AB}\right)\,,
\end{align}
and since $\hat{\rho}_{AB}$ is normalized we get for the first term on the rhs that
\begin{align}
P\left(\hat{\rho}_{AB},\bar{\rho}_{AB}\right)=\sqrt{1-\bar{F}^2\left(\hat{\rho}_{AB},\bar{\rho}_{AB}\right)}\,.
\end{align}
We continue with
\begin{align}
\bar{F}^2\left(\hat{\rho}_{AB},\bar{\rho}_{AB}\right)\geq\bar{F}^2\left(\tau_{AB},\bar{\rho}_{AB}\right)\geq\bar{F}^2\left(\tau_{ABC},\bar{\rho}_{ABC}\right)\,,
\end{align}
where the first step is Eq.~\eqref{eq:previous} and the second step follows since the fidelity is monotone under partial trace (this holds for general non-negative operators) together with choosing $\tau_{ABC}$ as an extension of $\tau_{AB}$ and $\bar{\rho}_{ABC}$ as an extension of $\bar{\rho}_{AB}$. We choose the purification of $\bar{\rho}_{AB}$ on $ABC$ defined through the pure state vector
\begin{align}
|\bar{\rho}_{ABC}\rangle:=\bar{\rho}_A^{\frac{1}{2}}|\Phi_{A:BC}\rangle\,,
\end{align}
where $|\Phi\rangle_{A:BC}$ denotes the non-normalized maximally entangled pure state vector in the cut $A:BC$ (on the subspace on $A$ spanned by the projector $P_A^\gamma$). Furthermore, we take the purification of $\tau_{AB}$ on $ABC$ given by
\begin{align}
|\tau_{ABC}\rangle:=\rho_A^{\frac{1}{2}}V_A\bar{\rho}_A^{-\frac{1}{2}}|\bar{\rho}_{ABC}\rangle.
\end{align}
In order to verify that it is a purification, we compute
\begin{align}
\tau_{AB}=\mathrm{Tr}_C\Big[|\tau_{ABC}\rangle\langle{\tau}_{ABC}|\Big]=\rho^{\frac{1}{2}}_AV_A\bar{\rho}_A^{-\frac{1}{2}}\bar{\rho}_{AB}\bar{\rho}_A^{-\frac{1}{2}}V_A^\dagger\rho^{\frac{1}{2}}_A\,.
\end{align}
We calculate
\begin{align}
&\bar{F}^2\left(\tau_{ABC},\bar{\rho}_{ABC}\right)=\left|\langle\bar{\rho}_{ABC}\middle|\tau_{ABC}\rangle\right|^2=\left|\langle\Phi_{A:BC}|\bar{\rho}_A^{\frac{1}{2}}|\tau_{ABC}\rangle\right|^2=\left|\langle\Phi_{A:BC}|\bar{\rho}_A^{\frac{1}{2}}\rho_A^{\frac{1}{2}}V_AP_A^\delta|\Phi_{A:BC}\rangle\right|^2\\
&=\left|\mathrm{Tr}\left[\bar{\rho}_A^{\frac{1}{2}}\rho_A^{\frac{1}{2}}V_AP_A^\delta\right]\right|^2=\left|\mathrm{Tr}\left[P_A^\delta\bar{\rho}_A^{\frac{1}{2}}\rho_A^{\frac{1}{2}}V_A\right]\right|^2=\left|\mathrm{Tr}\left[\bar{\rho}_A^{\frac{1}{2}}\rho_A^{\frac{1}{2}}V_A\right]\right|^2=\bar{F}^2\left(\bar{\rho}_A,\rho_A\right)=F^2\left(\bar{\rho}_A,\rho_A\right)\,.
\end{align}
Hence, together with Eq.~\eqref{eq:interest} we arrive at
\begin{align}\label{eq:arrive-at}
P\left(\hat{\rho}_{AB},\rho_{AB}\right)\leq P\left(\bar{\rho}_A,\rho_A\right)+P\left(\bar{\rho}_{AB},\rho_{AB}\right)\leq2\cdot P\left(\bar{\rho}_{AB},\rho_{AB}\right)\,,
\end{align}
where the last step follows from the monotonicity of the purified distance under partial trace. Using again the triangle inequality for the purified distance we then bound
\begin{align}
P\left(\bar{\rho}_{AB},\rho_{AB}\right)&\leq P\left(\bar{\rho}_{AB},\tilde{\rho}_{AB}\right)+P\left(\tilde{\rho}_{AB},\rho_{AB}\right)\\
&\leq P\left(P_A^\gamma\tilde{\rho}_{AB}P_A^\gamma,\tilde{\rho}_{AB}\right)+\eps\\
&\leq\sqrt{2\cdot\mathrm{Tr}\big[(1_A-P_A^\gamma)\tilde{\rho}_A\big]}+\eps\\
&\leq\sqrt{2\cdot\frac{\delta^2}{8}}+\eps=\frac{\delta}{2}+\eps\,,
\end{align}
where the third inequality follows from~\cite[Lem.~3.8]{mythesis}. Together with Eq.~\eqref{eq:arrive-at} we conclude that $P\left(\hat{\rho}_{AB},\rho_{AB}\right)\leq2\eps+\delta$.
\end{proof}

The standard asymptotic equipartition property for the max-divergence from~\cite[Thm.~6.3]{mybook} gives the asymptotic first-order expansion
\begin{align}\label{eq:Imax-asymptotic}
\lim_{n\to\infty}\frac{1}{n}I_{\max}^{\eps,P}(\dot{A}^n; B^n)_{\rho^{\otimes n}} = I(A \!:\! B)_\rho\,,
\end{align}
with the quantum mutual information defined as $I(A \!:\! B)_\rho := D(\rho_{AB}\|\rho_A\otimes\rho_B)$. We also find the following equivalence result for the smooth conditional min-entropy.

\begin{theorem}\label{thm:min-entropy}
Let $\rho_{AB}\in\cS_{\circ}(AB)$ and $0\leq2\eps+\delta\leq1$ with $\delta>0$. Then, we have\footnote{Similar equivalence results for alternative min-entropy definitions based on a maximization over $\sigma_B \in \cS_{\bullet}(B)$ can be derived by additionally employing~\cite[Lem.~21]{tomamichel10}.}
\begin{align}
H^{2\eps+\delta,P}_{\min}(A|\dot{B})_\rho\geq H_{\min}^{\eps,P}(A|B)_\rho-\log\frac{8+\delta^2}{\delta^2}\,,
\end{align}
and by definition we also have the opposite inequality $H_{\min}^{\eps,P}(A|\dot{B})_\rho\leq H_{\min}^{\eps,P}(A|B)_\rho$.
\end{theorem}

\begin{proof}
The first part of the proof is very similar to the proof of Thm.~\ref{lem:quantum_eq}, just with the roles of the systems $A$ and $B$ interchanged. In the following we only sketch the steps which are different. For $\tilde{\rho}_{AB}\in \cS_{\bullet}(AB)$ the optimizer in $H_{\min}^{\eps,P}(A|B)_\rho$ we define the bipartite quantum state
\begin{align}
\hat{\rho}_{AB}:=\rho^{\frac{1}{2}}_BV_B\bar{\rho}_B^{-\frac{1}{2}}\bar{\rho}_{AB}\bar{\rho}_B^{-\frac{1}{2}}V_B^\dagger\rho^{\frac{1}{2}}_B+\frac{1_A}{|A|}\otimes\left(\rho^{\frac{1}{2}}_B(1_B-V_BP^\gamma_BV_B^\dagger)\rho^{\frac{1}{2}}_B\right)\,,
\end{align}
with $P_B^\gamma$, $\bar{\rho}_{AB}$, and $V_B$ as in the proof of Thm.~\ref{lem:quantum_eq} (where $A\leftrightarrow B)$. We then find similarly as in the proof of Thm.~\ref{lem:quantum_eq} that
\begin{align}
\hat{\rho}_{AB}\leq\left(\frac{8}{\delta^2}\cdot\left\|\rho_B^{-\frac{1}{2}}\tilde{\rho}_{AB}\rho_B^{-\frac{1}{2}}\right\|_{\infty}+\frac{1}{|A|}\right)\cdot1_A\otimes\rho_B\,,
\end{align}
and using that $D_{\max}(\tilde{\rho}_{AB}\|1_A\otimes\rho_B)\geq-\log|A|$~\cite[App.~A]{tomamichel09} we get
\begin{align}
\frac{8}{\delta^2}\left\|\rho_B^{-\frac{1}{2}}\tilde{\rho}_{AB}\rho_B^{-\frac{1}{2}}\right\|_{\infty}+\frac{1}{|A|}\leq\frac{8+\delta^2}{\delta^2}\left\|\rho_B^{-\frac{1}{2}}\tilde{\rho}_{AB}\rho_B^{-\frac{1}{2}}\right\|_{\infty}\,.
\end{align}
As in the proof of Thm.~\ref{lem:quantum_eq} this leads to the statement
\begin{align}
H^{2\eps+\delta,P}_{\min}(A|\dot{B})\geq -D_{\max}(\tilde{\rho}_{AB}\|1_A\otimes\rho_B)-\log\frac{8+\delta^2}{\delta^2}\,,
\end{align}
concluding the proof.
\end{proof}

Employing the standard asymptotic equipartition property from~\cite[Thm.~7]{tomamichel08} this implies the asymptotic first-order expansion
\begin{align}\label{eq:Hmin-asymptotic}
\lim_{n\to\infty}\frac{1}{n}H_{\min}^{\eps,P}(A^n|\dot{B}^n)_{\rho^{\otimes n}} = H(A|B)_\rho\,,
\end{align}
with the conditional entropy defined as $H(A|B)_\rho:=-D(\rho_{AB}\|1_A\otimes \rho_B)$.


\section{Operational Examples}\label{sec:op}

\subsection{Overview}

It is generally neat that existing proofs and protocols readily apply and give tight bounds when combined with our novel restricted smoothing. In the following we discuss various basic classical and quantum examples in bipartite settings. 


\subsection{Classical state splitting}

Let $\eps \in (0,1]$ be the error parameter. There are two parties Alice and Bob. Alice possesses random variable $X$, taking values over a finite set $\mathcal{X}$ and  a random variable $Y$, taking values over a finite set $\mathcal{Y}$. Alice sends a message to Bob and at the end Bob outputs random variable $\hat{Y}$ such that $T(P_{XY},P_{X\hat{Y}})\leq \eps$. They are allowed to use shared randomness between them which is independent of $XY$ at the beginning of the protocol.

We note that a generalization of this task (with additional side information) was studied in \cite[Thm.~1]{Anshu17}. These results together with~\cite{BravermanR11} imply that the minimal number $R(P_{XY},\eps)$ of bits communicated from Alice to Bob to achieve classical state splitting with error $\eps\in(0,1]$ in generalized trace distance is bounded as
\begin{align}\label{eq:state-splitting_previous}
I_s^{\eps/(1-\delta)}(X; Y)_P - \log\frac{1}{\delta}\leq R(P_{XY},\eps)\leq I_s^{\eps - 3\delta}(X; Y)_P + 2\log\frac{1}{\delta}\,,
\end{align}
for $\delta \in (0,1)$ small enough. We show an even tighter characterization in terms of the smooth max-information.

\begin{theorem}\label{thm:state-splitting}
Let $P_{XY}\in\cS_{\circ}(XY)$. Then, the minimal number $R(P_{XY},\eps)$ of bits communicated from Alice to Bob to achieve classical state splitting with error $\eps\in(0,1]$ in generalized trace distance is bounded as
\begin{align}\label{eq:state-splitting}
I_{\max}^{\eps,T}(\dot{X};Y)_P\leq R(P_{XY},\eps)\leq I_{\max}^{\eps-\delta,T}(\dot{X};Y)_P+\log\log\frac{1}{\delta}+1,\quad\text{for any $\delta\in(0,\eps]$.}
\end{align}
In particular, this implies the asymptotic second-order expansion\footnote{Alternatively this expansion can also directly be deduced from Eq.~\eqref{eq:state-splitting_previous}.}
\begin{align}\label{eq:state-splitting2}
\frac{1}{n}R\left(P_{XY}^{\times n},\eps\right)=I(X;Y)_P+\sqrt{\frac{V(X;Y)_P}{n}}\cdot\Phi^{-1}(\eps)+O\left(\frac{\log n}{n}\right)\,.
\end{align}
\end{theorem}

Note that due to Thm.~\ref{thm:classical_equivalence} the bounds in Eq.~\eqref{eq:state-splitting} can be further bounded in terms of the information spectrum max-information. Whereas the lower bound then basically becomes the same as in Eq.~\eqref{eq:state-splitting_previous}, the upper bounds improves the fudge term $O\left(\log\frac{1}{\delta}\right)$ from Eq.~\eqref{eq:state-splitting_previous} to $O\left(\log\log\frac{1}{\delta}\right)$.

\begin{proof}
Eq.~\eqref{eq:state-splitting2} immediately follow from Eq.~\eqref{eq:state-splitting} and Thm.~\ref{thm:classical_equivalence}. The converse in Eq.~\eqref{eq:state-splitting} is as follows, which uses the converse argument from ~\cite[Thm.~2]{Anshu17}. Let $T$ be Alice's message and $S$ be shared randomness. Observe that $P_{XS} = P_{X}\times P_S$, let $\mathcal{D}: ST \rightarrow Y$ be Bob's decoding operation, $U_T$ be the random variable uniformly distributed over $T$, and let the output random variable after Alice's message and Bob's decoding be $P'_{XY}:= (\id_X\times\mathcal{D})(P_{XST})$. Now, we consider that
\begin{align}
R &\geq D_{\max}(P_{XST}\| P_{XS}\times U_T) = D_{\max}(P_{XST}\| P_{X}\times P_S\times U_T)\nonumber\\
&\geq D_{\max}((\id_X\times\mathcal{D})(P_{XST})\| P_{X}\times \mathcal{D}(P_S\times U_T))\nonumber\\
&=D_{\max}(P'_{XY}\| P_X \times \mathcal{D}(P_S\times U_T)) \geq \min_{Q_Y}D_{\max}(P'_{XY}\| P_X \times Q_Y)\,.
\end{align}
We then use the fact that $P'_X=P_X$ and $T(P'_{XY}, P_{XY}) \leq \eps$ to further lower bound $R$ by $I_{\max}^{\eps,T}(\dot{X};Y)_P$.\\

The achievability in Eq.~\eqref{eq:state-splitting} uses the rejection sampling argument \cite{jain02, HarshaJMR10}. Let $P'_{XY}, Q_Y$ be random variables achieving the infimum in the definition of $I_{\max}^{\eps - \delta}(\dot{X};Y)_P$. Let $K:= I_{\max}^{\eps - \delta}(\dot{X};Y)_P$ and $R:= K+ \log\log\frac{1}{\delta}$. By definition, it holds that $P'_{XY} \leq 2^K P_X \times Q_Y$ and $P'_X=P_X$. Thus we conclude that for all $x$ satisfying $P_X(x)>0$, $P'_{Y\mid X=x} \leq 2^K Q_Y$.
 
\vspace{0.1in}

{\it The protocol $\cP$}: Alice and Bob share the random variable $P_X \times Q_{Y_1}\times \ldots Q_{Y_{2^R}}$, with $X$ belonging to Alice and $Y_1, \ldots Y_{2^R}$ acting as shared randomness between Alice and Bob. They proceed with the following steps, with Alice obtaining a sample $x$ from $P_X(x)$.
\begin{enumerate}
\item Alice sets $i=1$.
\item (While $i\leq 2^R$):
\item Alice takes a sample $y$ from $Q_{Y_i}$. 
\item With probability $\frac{P'_{Y\mid X=x}(y)}{2^K Q_Y(y)}$ she accepts this sample, sends $i$ to Bob and exits the while loop.
\item With probability $1-\frac{P'_{Y\mid X=x}(y)}{2^K Q_Y(y)}$ she updates $i\rightarrow i+1$ and goes to Step $2$ (End While). 
\item If $i> 2^R$, Alice sends $2^R+1$ to Bob.  
\item Bob receives Alice's message, which we call $j$. If $j> 2^R$, Bob outputs a sample distributed as $Q_Y$. Else he outputs the sample from $Q_{Y_j}$. 
\end{enumerate}
Let the output of Bob be $P''_{Y\mid X=x}$.

\vspace{0.1in}

{\it Analysis of the protocol}: The probability of Alice's acceptance on Step $4$ is $$\sum_y Q_Y(y)\frac{P'_{Y\mid X=x}(y)}{2^K Q_Y(y)} = 2^{-K}\sum_y P'_{Y\mid X=x}(y) = 2^{-K}.$$
Conditioned on Alice's acceptance, the distribution of $Y_i$ is equal to $P'_{Y\mid X=x}$. To argue this, observe that the probability of any $y$, conditioned on acceptance, is equal to $$\frac{1}{2^{-K}}\cdot Q_Y(y)\cdot\frac{P'_{Y\mid X=x}(y)}{2^K Q_Y(y)} = P'_{Y\mid X=x}(y).$$
Thus, conditioned on the event that Alice accepts an $i$, the sample output by Bob (which is the same as that observed by Alice) is distributed as $P'_{Y\mid X=x}$. Let $\gamma$ be the probability that Alice does not find any sample, that is, $i> 2^R$. Then Bob's output $P''_{Y\mid X=x}$ is equal to $(1-\gamma)P'_{Y\mid X=x} + \gamma Q_{Y}$. Let $P''_{XY}:= P_XP''_{Y\mid X}$ be the overall output distribution. 

Since probability of acceptance at any step is equal to $2^{-K}$, we have 
$$\gamma = (1- 2^{-K})^{2^R} \leq \left(2^{-2^{-K}}\right)^{2^R} = 2^{-2^{R-K}} = 2^{- 2^{\log\log\frac{1}{\delta}}}= 2^{- \log\frac{1}{\delta}} = \delta.$$ 
Consider,
\begin{eqnarray*}
T(P''_{XY}, P_{XY}) &=& \sum_xP_X(x)T(P''_{Y\mid X=x}, P_{Y\mid X=x})\\ 
&=& \sum_xP_X(x)T((1-\gamma)P'_{Y\mid X=x} + \gamma Q_{Y}, P_{Y\mid X=x}))\\
&\leq& (1-\gamma)\sum_xP_X(x)T(P'_{Y\mid X=x}, P_{Y\mid X=x})) + \gamma\sum_xP_X(x)T(Q_{Y}, P_{Y\mid X=x})\\
&\leq& (1-\gamma) T(P'_{XY}, P_{XY}) + \gamma \leq \eps-\delta + \gamma \leq \eps.
\end{eqnarray*}
Furthermore, the number of bits communicated is $\log(2^R+1)\leq R+1$, which completes the proof.
\end{proof}


\subsection{Strong privacy amplification against side information}

The goal of privacy amplification against side information is to extract, from a random variable $X$ correlated with (quantum) side information $B$, a uniformly distributed random variable $Z$ that is independent of $B$. Common privacy amplification schemes use a uniformly random and independent seed $S$ to apply a function $\{f_{X\to Z}^s\}_{s\in S}$ to $X$, and they are called strong if the seed remains independent of the side information after the function is applied.

For a set of functions $\{f_{X\to Z}^s\}_{s\in S}$ and a classical-quantum state
\begin{align}
\rho_{XB}=\sum_{X\in\mathcal{X}}|x\rangle\langle x|\otimes\rho_B^x\in\cS_\circ(XB)
\end{align}
we use the same composable security criterion for $\eps$-random and secret bits as, e.g, in~\cite[Sect.~7.3]{mybook}, namely we say that the resulting random variable $Z$ is \emph{$\eps$-random and secret} if
\begin{align}\label{eq:security-criterion}
\Delta \Bigg(\underbrace{\frac{1}{|S|}\sum_{\substack{s\in S\\z\in\mathcal{Z}}}|s\rangle\langle s|_S\otimes|z\rangle\langle z|\otimes\Bigg(\sum_{x:f^s(x)=z}\rho_B^x\Bigg)}_{=:\;\omega_{SZB}},\frac{1_S}{|S|}\otimes\frac{1_Z}{|Z|}\otimes\rho_B\Bigg)\leq\eps\,.
\end{align}
Note that, in contrast to the setting studied in~\cite[Sect. III]{tomamichel12} or~\cite{tomamichel10}, we have a composable security definition by putting the reduced state on $B$ on the lhs of Eq.~\eqref{eq:security-criterion}. We refer to~\cite[App.~B]{portmann14} for a more detailed discussion. Our goal is now to find bounds on the maximal number of $\eps$-random and secret bits, denoted $\ell^{\Delta}(\rho_{XB},\eps)$ for $\Delta \in \{T, P\}$, that can be extracted from $\rho_{XB}$ by any privacy amplification scheme of the above general form.

\begin{theorem}\label{thm:pa-qsi}
Let $\rho_{XB}\in\cS_{\circ}(XB)$ be classical-quantum on $XB$ and $\eps\in(0,1]$. Then, 
\begin{align}\label{eq:pa-qsi}
H_{\min}^{\eps-\delta,P}(X|\dot{B})_\rho-\log\frac{1}{\delta^4}\leq\ell^{P}(\rho_{XB},\eps)\leq H_{\min}^{\eps,P}(X|\dot{B})_\rho,\quad\text{for any $\delta\in(0,\eps]$.}
\end{align}
Moreover, when $B = Y$ is classical then we also have
\begin{align}\label{eq:pa-cl}
H_{\min}^{\eps-\delta,T}(X|\dot{Y})_\rho-\log\frac{1}{4\delta^2}\leq\ell^{T}(\rho_{XY},\eps)\leq H_{\min}^{\eps,T}(X|\dot{Y})_\rho\,.
\end{align}
\end{theorem}

The upper bound in Eq.~\eqref{eq:pa-qsi} is tighter than previously known bounds since only partial smoothing is allowed and the marginal $\rho_B$ is fixed, in contrast to~\cite[Sect. III]{tomamichel12}. This is crucial, as at least in the classical special case this allows us to recover tight second-order bounds. The lower bound is incomparable to existing results and provided solely to show that it is equivalent to the upper bound up to terms that vanish in second order in the number of copies $n$.
The bound for classical side information in Eq.~\eqref{eq:pa-cl} in particular implies the asymptotic second-order expansion
\begin{align}
\frac{1}{n}\ell^T \left(P_{XY}^{\times n},\eps\right)=H(X|Y)_P+\sqrt{\frac{V(X|Y)_P}{n}}\cdot\Phi^{-1}(\eps)+O\left(\frac{\log n}{n}\right)
\end{align}
as first given in~\cite[Thm.~25]{hayashi16} (see also~\cite[Thm.~3]{watanabe12}).

\begin{proof}
We first prove the lower bound in Eq.~\eqref{eq:pa-qsi}. Let $\tilde{\rho}_{XB}\in\cS_{\bullet}(XB)$ be the optimizer in the definition of $H_{\min}^{\eps-\delta,P}(X|\dot{B})_\rho$ and let
\begin{align}\label{eq:pa_tripartite}
\tilde{\omega}_{SZB}:=\frac{1}{|S|}\sum_{\substack{s\in S\\z\in\mathcal{Z}}}|s\rangle\langle s|_S\otimes|z\rangle\langle z|\otimes\Bigg(\sum_{x:f^s(x)=z}\tilde{\rho}_B^x\Bigg)\,.
\end{align}
Since by definition $\rho_B\geq\tilde{\rho}_B$ and by data-processing $P\left(\omega_{SZB},\tilde{\omega}_{SZB}\right)\leq\eps-\delta$ we get that
\begin{align}
P\left(\omega_{SZB},\frac{1_S}{|S|}\otimes\frac{1_Z}{|Z|}\otimes\rho_B\right)&\leq P\left(\omega_{SZB},\frac{1_S}{|S|}\otimes\frac{1_Z}{|Z|}\otimes\tilde{\rho}_B\right)\\
&\leq P\left(\tilde{\omega}_{SZB},\frac{1_S}{|S|}\otimes\frac{1_Z}{|Z|}\otimes\tilde{\rho}_B\right)+P\left(\omega_{SZB},\tilde{\omega}_{SZB}\right)\\
&\leq P\left(\tilde{\omega}_{SZB},\frac{1_S}{|S|}\otimes\frac{1_Z}{|Z|}\otimes\tilde{\rho}_B\right)+\eps-\delta\\
&\leq \sqrt{2T\left(\tilde{\omega}_{SZB},\frac{1_S}{|S|}\otimes\frac{1_Z}{|Z|}\otimes\tilde{\rho}_B\right)}+\eps-\delta\,,
\end{align}
where the first step is based on the operator monotonicity of the square root, and in the last step we employed the equivalence of generalized trace distance and purified distance from Eq.~\eqref{eq:pt-equivalence}. Now, standard achievability proofs such as~\cite[Thm.~6]{tomamichel10} applied to $\tilde{\rho}_{XB}\in\cS_{\bullet}(XB)$ give
\begin{align}
T\left(\tilde{\omega}_{SZB},\frac{1_S}{|S|}\otimes\frac{1_Z}{|Z|}\otimes\tilde{\rho}_B\right)\leq\frac{1}{2}\sqrt{|Z|\cdot2^{-H_{\min}^{\eps-\delta,P}(X|\dot{B})_\rho}}\,,
\end{align}
and choosing $\log|Z|=H_{\min}^{\eps-\delta,P}(X|\dot{B})_\rho-\log\frac{1}{\delta^4}$ leads to the claim.

For the upper bound in Eq.~\eqref{eq:pa-qsi} we follow~\cite[Sect.~7.3.3]{mybook} but adapted to our partially smoothed conditional min-entropy. Namely, assume by contradiction that there exists a protocol which extracts $\ell>H_{\min}^{\eps,P}(X|\dot{B})_\rho$ bits of $\eps$-random and secret bits. Then, since applying a function on $X$ cannot increase the smooth conditional min-entropy (Lem.~\ref{lem:classical-function}) we have for all $s\in\cS$ that
\begin{align}
\ell>H_{\min}^{\eps,P}(X|\dot{B})_\rho\geq H_{\min}^{\eps,P}(Z|\dot{B})_{\rho^s},\quad\text{with $\rho_{ZB}^s:=\sum_{z\in\mathcal{Z}}|z\rangle\langle z|_Z\otimes\left(\sum_{x:f^s(x)=z}\rho_B^x\right)$.}
\end{align}
Hence, for all $\tilde{\rho}_{ZB}\in\cS_{\bullet}(ZB)$ with $P\left(\tilde{\rho}_{ZB},\rho^s_{ZB}\right)\leq\eps$ we have $H_{\min}(Z|B)_{\tilde{\rho}}<\ell$. This in turn implies
\begin{align}
P\left(\rho_{ZB}^s,\frac{1_Z}{|Z|}\otimes\rho_B\right)>\eps\quad\implies\quad P\left(\omega_{SZB},\frac{1_S}{|S|}\otimes\frac{1_Z}{|Z|}\otimes\rho_B\right)>\eps\,,
\end{align}
which is in contradiction to Eq.~\eqref{eq:security-criterion} and thus concludes the proof.

The upper bound in Eq.~\eqref{eq:pa-cl} follows in the same way as the upper bound in Eq.~\eqref{eq:pa-qsi}, just by noting that in the classical case Lem.~\ref{lem:classical-function} also holds for the generalized trace distance. For the lower bound in Eq.~\eqref{eq:pa-cl}, denote in the security criterion Eq.~\eqref{eq:security-criterion} the state $\omega_{SZB}$ for $B=Y$ classical by $Q_{SZY}$, let $\tilde{P}_{XY}\in\cS_{\bullet}(XY)$ be the optimizer in the definition of $H_{\min}^{\eps-\delta,T}(X|\dot{Y})_P$, and let $\tilde{Q}_{SZY}$ be defined as in Eq.~\eqref{eq:pa_tripartite}. Since by definition $P_Y\geq\tilde{P}_Y$ and by data-processing $T\left(Q_{SZY},\tilde{Q}_{SZY}\right)\leq\eps-\delta$ we get that
\begin{align}
T\left(Q_{SZY},\frac{1_S}{|S|}\times\frac{1_Z}{|Z|}\times P_Y\right)&\leq T\left(\tilde{Q}_{SZY},\frac{1_S}{|S|}\times\frac{1_Z}{|Z|}\times\tilde{P}_Y\right)+T\left(Q_{SZY},\tilde{Q}_{SZY}\right)\\
&\leq T\left(\tilde{Q}_{SZY},\frac{1_S}{|S|}\times\frac{1_Z}{|Z|}\times\tilde{P}_Y\right)+\eps-\delta\,.
\end{align}
Now, standard achievability proofs such as~\cite[Thm.~6]{tomamichel10} applied to $\tilde{P}_{XY}\in\cS_{\bullet}(XY)$ lead to the claim for $\log|Z|=H_{\min}^{\eps-\delta,T}(X|\dot{Y})_P-\log\frac{1}{4\delta^2}$.
\end{proof}

For the case of quantum side information we do not have the asymptotic second-order expansion of $H_{\min}^{\eps,P}(X|\dot{B})_\rho$ and thus we cannot give the asymptotic second-order expansion of $\ell(\rho_{XB},\eps)$. This seems to be an open problem for the composable security definition used here (see~\cite{hayashi16-2} for a discussion of the Markovian case). However, note that by Thm.~\ref{thm:pa-qsi} finding the asymptotic second-order expansion of $H_{\min}^{\eps,P}(X|\dot{B})_\rho$ is now equivalent to finding the asymptotic second-order expansion of $\ell(\rho_{XB},\eps)$. Hence, we believe that our ideas provide a promising approach to study the question.


\subsection{Quantum state merging}

A pure tripartite state $\rho_{ABR}$ is shared between parties Alice (A), Bob (B), and the reference $R$. The goal is to send the $A$-marginal from Alice to Bob using classical communication and entanglement assistance while not changing the overall state~\cite{horodecki05,horodecki06,berta08,dupuis10}. More precisely, for $\rho_{ABR}\in\cS_\circ(ABR)$ of rank-one and $A_{0}B_{0}$ additional quantum systems, a quantum channel
\begin{align}
\mathcal{E}:AA_{0}\otimes BB_{0}\rightarrow A_{1}\otimes B_{1}\bar{B}B
\end{align}
is a quantum state merging of $\rho_{ABR}$ with error $\eps\in[0,1]$, if it is a local operation and classical forward communication process for the bipartition $AA_{0}\rightarrow A_{1}$ versus $BB_{0}\rightarrow B_{1}\bar{B}B$, and
\begin{align}
P\Big((\mathcal{E}_{AA_{0}BB_{0}\to A_{1}B_{1}\bar{B}B}\otimes\mathcal{I}_R)(\Phi_{A_{0}B_{0}}\otimes\rho_{ABR}),\Phi_{A_{1}B_{1}}\otimes\rho_{B\bar{B}R}\Big)\leq\eps\,,
\end{align}
where $\rho_{B\bar{B}R}=(\mathcal{I}_{A\to\bar{B}}\otimes\mathcal{I}_{BR})(\rho_{ABR})$, and $\Phi_{A_{0}B_{0}}$, $\Phi_{A_{1}B_{1}}$ are maximally entangled states on $A_{0}B_{0}$, $A_{1}B_{1}$, respectively. The difference $\log|A_0|-\log|A_1|$ quantifies the entanglement cost.

\begin{theorem}
Let $\rho_{ABR}\in\cS_\circ(ABR)$ be a pure state. For free classical communication assistance the minimal entanglement cost $E(\rho_{ABR},\eps)$ for quantum state merging of $\rho_{ABR}$ with error $\eps\in(0,1]$ in purified distance is bounded as
\begin{align}\label{eq:entanglement-cost}
-H_{\min}^{\eps,P}(A|\dot{R})_\rho\leq E(\rho_{ABR},\eps)\leq -H_{\min}^{\eps-\delta,P}(A|\dot{R})_\rho+\log\frac{1}{\delta^4},\quad\text{for any $\delta\in(0,\eps]$.}
\end{align}
Alternatively, for unlimited entanglement assistance\,---\,not necessarily constraint to the form of maximally entangled states\,---\,the minimal classical communication cost $C(\rho_{ABR},\eps)$ for quantum state merging of $\rho_{ABR}$ with error $\eps\in(0,1]$ in purified distance is bounded as
\begin{align}\label{eq:classical-cost}
I_{\max}^{\eps,P}(\dot{R};A)_\rho\leq C(\rho_{ABR},\eps)\leq I_{\max}^{\eps-\delta,P}(\dot{R};A)_\rho+\log\frac{1}{\delta^4},\quad\text{for any $\delta\in(0,\eps]$.}
\end{align}
\end{theorem}

Note that the bounds in Eq.~\eqref{eq:entanglement-cost} and Eq.~\eqref{eq:classical-cost} improve on previously known bounds on one-shot quantum state merging (see, e.g., \cite{berta08,Anurag17,Majenz17} and follow-up works). For the converse direction, our bounds in terms of partially smoothed entropies are tighter because globally smoothed entropies allow a larger neighbourhood to smooth over. For the achievability direction, previous bounds all featured smoothing parameters of at least $2\eps$ for $\eps$-error protocols (which does not seem to be second-order tight). The asymptotic second-order expansions of $E$ and $C$ are an open problem but are now reduced to giving the asymptotic second-order expansions of $H_{\min}^{\eps,P}(A|\dot{R})_\rho$ and $I_{\max}^{\eps,P}(\dot{R};A)_\rho$, respectively.

\begin{proof}
For the lower bounds in Eq.~\eqref{eq:entanglement-cost} and Eq.~\eqref{eq:classical-cost} we first model the form of a general quantum state merging protocol.

On Alice's side, we consider an arbitrary local operation from $A A_0$ to a classical register $X_{A}$ (to be sent to Bob) and $A_1$. We consider an isometric purification of this operation in two steps. First, Alice performs an isometry $U$ from $AA_{0}$ to $A_{1}X_{A}A'$, where $X_{A}$ is a register that is to be measured and $A'$ an arbitrary garbage register to be discarded. Then, the measurement of $X_{A}$ and the communication to Bob is modelled by the isometry $V=\sum_{x}|xxx\rangle_{X_{A}X_{B}X_{R}}\langle x|_{X_{A}}$ from $X_{A}$ to $X_{A}X_{B}X_{R}$, where $X_B$ is a classical copy of $X_A$ to be sent to Bob and $X_R$ is a coherent copy of $X_A$ and $X_B$ that is used to purify this operation.

On Bob's side, we consider an arbitrary local operation from $B_{0}X_{B}B$ to $B_{1}\bar{B}BX_{B}$, where $\bar{B}B$ is the merged state. This can again be purified to an isometry $W$ from $B_{0}BX_{B}$ to $B_{1}\bar{B}BX_{B}B'$, with $B'$ an arbitrary register to be discarded. Overall, the pure state vector $|\Phi\rangle_{A_{0}B_{0}}\otimes|\rho\rangle_{ABR}$ is taken to some pure state vector
\begin{align}
\text{$|\omega\rangle_{A_{1}B_{1}\bar{B}BRX_{A}X_{B}X_{R}A'B'}$ $\eps$-close in purified distance to $|\Phi\rangle_{A_{1}B_{1}}\otimes|\rho\rangle_{\bar{B}BR}\otimes|\xi\rangle_{X_{A}X_{B}X_{R}A'B'}$,}
\end{align}
where $|\xi\rangle_{X_{A}X_{B}X_{R}A'B'}$ denotes another pure state vector. By basic properties of the purified distance~\cite[Chap.~3]{mythesis}, this latter vector has without loss of generality the form
\begin{align}
|\xi\rangle_{X_{A}X_{B}X_{R}A'B'}=\sum_{x}\sqrt{p_{x}}|xxx\rangle_{X_{A}X_{B}X_{R}}\otimes|\xi^{x}\rangle_{A'B'}\,,
\end{align}
where $\{p_{x}\}$ denotes some probability distribution and the $|\xi^{x}\rangle_{A'B'}$ are pure state vectors on $A'B'$.

For the lower bound in Eq.~\eqref{eq:entanglement-cost} we follow~\cite[Sect.~4.2]{berta08} and analyse the correlations between Alice and the reference system, measured in terms of the smooth conditional min-entropy. For the modelling as above we estimate
\begin{align}
\log|A_0|+H_{\min}^{\eps,P}(A|\dot{R})_{\rho}&\geq H_{\min}^{\eps,P}(A_{0}A|\dot{R})_{\Phi\otimes\rho}\label{eq:dimension-upper}\\
&=H_{\min}^{\eps,P}(A_{1}A'X_{A}|\dot{R})_{U(\Phi\otimes\rho)U^{\dagger}}\label{eq:unintary-invariance}\\
&\geq H_{\min}(A_{1}A'X_{A}|R)_{V^\dagger(\Phi\otimes\rho\otimes\xi)V}\\
&=-D_{\max}\left(V^\dagger\left(\frac{1_{A_1}}{|A_1|}\otimes\rho_{R}\otimes\xi_{X_{A}X_{B}X_{R}A'}\right)V\middle\|1_{A_1A'X_A}\otimes\rho_{R}\right)\\
&\geq -D_{\max}\left(\frac{1_{A_1}}{|A_1|}\otimes\rho_{R}\otimes\xi_{X_{A}X_{R}A'}\middle\|1_{A_1A'X_A}\otimes\xi_{X_{R}}\otimes\rho_{R}\right)\label{eq:projective-measurement}\\
&=\log|A_1|-D_{\max}(\xi_{X_{A}X_{R}A'}\|1_{A'X_A}\otimes\xi_{X_{R}})\\
&\geq\log|A_1|\label{eq:classical-quantum}\,,
\end{align}
where in Eq.~\eqref{eq:dimension-upper} we used the dimension upper bound for the smooth conditional min-entropy from Lem.~\ref{lem:tripartite-dimension}, in Eq.~\eqref{eq:unintary-invariance} the isometric invariance of the smooth conditional min-entropy (Lem.~\ref{lm:iso}), in Eq.~\eqref{eq:projective-measurement} the monotonicity property under projective measurements when conditioned on the measurement outcomes from Lem.~\ref{lem:projective-measurement}, and in Eq.~\eqref{eq:classical-quantum} that $D_{\max}(\xi_{X_{A}X_{R}A'}\|1_{A'X_A}\otimes\xi_{X_{R}})\leq 0$ following from the classical-quantum structure of $\xi_{X_{A}X_{R}A'}$~\cite[Lem.~3.1.9]{renner05}. 

For the lower bound in Eq.~\eqref{eq:classical-cost} we analyse the correlations between Alice and the reference system, measured in terms of the smooth max-information. For the modelling as above we find
\begin{align}
I_{\max}^{\eps,P}(\dot{R};A)_{\rho}&\leq I_{\max}^{\eps,P}(\dot{R};A_{0}A)_{\Phi\otimes\rho}\\
&=I_{\max}^{\eps,P}(\dot{R};A_{1}A'X_{A}X_{B}X_{R})_{\omega}\\
&\leq I_{\max}^{\eps,P}(\dot{R};A_{1}A'X_{A}X_{R})_{\omega}+\log|X|\\
&\leq\log|X|\,,
\end{align}
where the first step is due to the monotonicity of the smooth max-information under local quantum operations (Lem.~\ref{lm:cptp}), the second step due to the invariance of the smooth max-information under isometries (Lem.~\ref{lm:iso}), the third step due to the dimension upper bound on the smooth max-information of coherent classical states from Lem.~\ref{lem:coherent}, and the forth step follows because the output state has to be $\eps$-close in purified distance to the perfect state (which has no correlations to $R$). Note that this chain of arguments does not depend on the structure of the entanglement assistance and thus also applies to assistance that is not necessarily constraint to the form of maximally entangled states.

The upper bound in Eq.~\eqref{eq:entanglement-cost} follows from the analysis in~\cite[Prop.~4.7]{berta08} adapted to our partially smoothed conditional min-entropy. The protocol is such that Alice applies a Haar random rank-$|A_1|$ projective measurement to decouple her systems from the reference, sends the resulting classical measurement outcomes to Bob, who then recovers the full state by Uhlmann's theorem. In particular, fix $N$ orthogonal subspaces of dimension $|A_1|$ on $AA_0$, denote the projectors on these subspaces followed by a fixed unitary mapping it to $A_1$ by $P^x_{A_0A\to A_1}$, and define the isometry
\begin{align}
W_{A_{0}A\to A_{1}X_{A}X_{B}}=\sum_{x}P^{x}_{A_{0}A\to A_{1}}\otimes|x\rangle_{X_{A}}\otimes|x\rangle_{X_{B}}\,.
\end{align}
Now, by standard one-shot decoupling results as for example outlined in~\cite[Thm.~5.2]{dupuis10}, there exists a unitary operator $U_{A_{0}A}$ such that for
\begin{align}
\omega_{A_{1}X_{A}X_{B}B_{0}BR}=W_{A_{0}A\to A_{1}X_{A}X_{B}}U_{A_{0}A} \left(\Phi_{A_0B_0} \otimes \rho_{ABR}\right) U_{A_{0}A}^{\dagger}W^\dagger_{A_{0}A\to A_{1}X_{A}X_{B}}\,,
\end{align}
we have
\begin{align}
T\left(\omega_{A_1X_AR},\tau_{A_1X_A}\otimes\rho_R\right)\leq\frac{1}{2}\sqrt{|A_1|\cdot2^{-H_{\min}(A_0A|R)_{\Phi\otimes\rho}}}=\frac{1}{2}\sqrt{\frac{|A_1|}{|A_0|}\cdot2^{-H_{\min}(A|R)_\rho}}\,,
\end{align}
where $\tau_{A_1X_AX_B}=W_{A_{0}A\to A_{1}X_{A}X_{B}}\left(\frac{1_{A_0}}{|A_0|}\otimes\frac{1_{X_A}}{|X_A|}\right)W^\dagger_{A_{0}A\to A_{1}X_{A}X_{B}}$ and we have used the additivity of the conditional min-entropy~\cite{koenig08}. By the equivalence of generalized trace distance and purified distance from Eq.~\eqref{eq:pt-equivalence} this implies
\begin{align}
P\left(\omega_{A_1X_AR},\tau_{A_1X_A}\otimes\rho_R\right)\leq\left(\frac{|A_1|}{|A_0|}\cdot2^{-H_{\min}(A|R)_\rho}\right)^{\frac{1}{4}}\,.
\end{align}
Moreover, by Uhlmann's theorem there exists an isometry $V_{BB_{0}X_{B}\to BB'B_{1}X_{B}}$ with
\begin{align}\label{eq:uhlmann-merging}
&P\left(\omega_{A_1X_AR},\tau_{A_1X_A}\otimes\rho_R\right)\notag\\
&=P\left(V_{BB_{0}X_{B}\rightarrow BB'B_{1}X_{B}}\big(\omega_{A_{1}X_{A}X_{B}B_{0}BR}\big)V^\dagger_{BB_{0}X_{B}\to BB'B_{1}X_{B}},\tau_{X_AX_B}\otimes\Phi_{A_1B_1}\otimes\rho_{BB'R}\right)\,.
\end{align}
We conclude that applying the isometry $W_{A_{0}A\to A_{1}X_{A}X_{B}}$, sending $X_B$ to Bob, and then applying the isometry $V_{BB_{0}X_{B}\rightarrow BB'B_{1}X_{B}}$, realises quantum state merging for an
\begin{align}
\text{entanglement cost}\quad\log|A_0|-\log|A_1|\quad\text{and error}\quad\left(\frac{|A_1|}{|A_0|}\cdot2^{-H_{\min}(A|R)_\rho}\right)^{\frac{1}{4}}\,.
\end{align}
Now, let $\tilde{\rho}_{AR}\in\cS_{\bullet}(AR)$ be the optimizer in $H_{\min}^{\eps-\delta,P}(A|\dot{R})_\rho$ and choose $A_0B_0$ and $A_1B_1$ such that
\begin{align}\label{eq:choice-merging}
\log|A_0|-\log|A_1|=-H_{\min}^{\eps-\delta,P}(A|\dot{R})_\rho+\log\frac{1}{\delta^4}\,.
\end{align}
For the quantum state
\begin{align}
\tilde{\omega}_{A_{1}X_{A}X_{B}B_{0}BR}:=W_{A_{0}A\to A_{1}X_{A}X_{B}}U_{A_{0}A} \left(\Phi_{A_0B_0} \otimes \tilde{\rho}_{ABR}\right) U_{A_{0}A}^{\dagger}W^\dagger_{A_{0}A\to A_{1}X_{A}X_{B}}
\end{align}
with $\tilde{\rho}_{ABR}$ a purification of $\tilde{\rho}_{AR}$ we estimate
\begin{align}
P\left(\omega_{A_1X_AR},\tau_{A_1X_A}\otimes\rho_R\right)&\leq P\left(\omega_{A_1X_AR},\tau_{A_1X_A}\otimes\tilde{\rho}_R\right)\label{eq:first}\\
&\leq P\left(\tilde{\omega}_{A_1X_AR},\tau_{A_1X_A}\otimes\tilde{\rho}_R\right)+P\left(\tilde{\omega}_{A_1X_AR},\omega_{A_1X_AR}\right)\label{eq:second-triangle}\\
&\leq P\left(\tilde{\omega}_{A_1X_AR},\tau_{A_1X_A}\otimes\tilde{\rho}_R\right)+\eps-\delta\label{eq:third-data}\\
&\leq\left(\frac{|A_1|}{|A_0|}\cdot2^{-H_{\min}^{\eps-\delta,P}(A|\dot{R})_\rho}\right)^{\frac{1}{4}}+\eps-\delta\\
&\leq\eps\,,\label{eq:last}
\end{align}
where in Eq.~\eqref{eq:first} we used that by definition $\rho_R\geq\tilde{\rho}_R$, in Eq.~\eqref{eq:second-triangle} the triangle inequality, in Eq.~\eqref{eq:third-data} that by data-processing $P\left(\omega_{A_{1}X_{A}R},\tilde{\omega}_{A_{1}X_{A}R}\right)\leq\eps-\delta$, and in Eq.~\eqref{eq:last} we applied the choice from Eq.~\eqref{eq:choice-merging}. By Uhlmann's theorem as in Eq.~\eqref{eq:uhlmann-merging} this concludes the proof.

The upper bound in Eq.~\eqref{eq:classical-cost} follows from the protocol given in~\cite[Thm.~1]{Anurag17} based on the convex split lemma (cf.~Lem.~\ref{convsplitlemm}).\footnote{Analogously to the protocol minimizing the entanglement cost, this protocol can also be thought of in terms of decoupling~\cite{Majenz17}.} We sketch the argument in our context. Given $\tilde{\rho}_{AR}\in\cS_{\circ}(AR)$ and $\sigma_A\in\cS_{\circ}(A)$ the optimizers in $I_{\max}^{\eps-\delta,P}(\dot{R};A)_\rho$, the protocol from~\cite[Thm.~1]{Anurag17} together with quantum teleportation realises a $\delta$-error quantum state merging of a purification $\tilde{\rho}_{ABR}$ with $P(\tilde{\rho}_{ABR},\rho_{ABR})\leq\eps-\delta$, for a classical communication cost of $D_{\max}\left(\tilde{\rho}_{AR}\middle\|\sigma_A\otimes\rho_R\right)+\log\frac{1}{\delta^4}$. Following the same line of argument as in Eq.~\eqref{eq:first} - \eqref{eq:last} based on $\rho_R=\tilde{\rho}_R$ and $P(\tilde{\rho}_{ABR},\rho_{ABR})\leq\eps-\delta$ this leads to the desired statement.
\end{proof}

The asymptotic first-order expansions then follow from Eq.~\eqref{eq:Imax-asymptotic} and Eq.~\eqref{eq:Hmin-asymptotic} and we recover the original results on quantum state merging~\cite{horodecki05,horodecki06}. As shown in these references in first-order asymptotically the entanglement cost and classical communication cost can actually be simultaneously minimised\,---\,whereas this becomes unclear in the one-shot setting.


\section{Outlook}\label{sec:outlook}

As we have seen our locally smoothed information measures naturally appear in a plethora of operational tasks in quantum information theory. It might be insightful to study mathematical properties of these measures that go beyond what we presented in \cref{sec:def}. The main open problem raised by our work is to give asymptotic second-order expansions of the partially smoothed information measures
\begin{align}
H_{\min}^{\eps,\Delta}(A|\dot{B})_\rho\quad\text{and}\quad I_{\max}^{\eps,\Delta}(\dot{A};B)_\rho
\end{align}
for the quantum case. This, however, seems to require new ideas as the classical proof technique from Thm.~\ref{thm:classical_equivalence} does not directly translate to the quantum setting and the quantum proof technique from Thm.~\ref{lem:quantum_eq} is not tight enough.

Locally smoothed information measures also appear naturally when defining smooth entropies for quantum channels, as realised in the recent works~\cite{Fang18,Faist18}. We especially point to~\cite{Fang18} where our equivalence results from \cref{sec:equi} already found applications in the context of quantum channel simulations. Finally, another question our approach might shine some light on is the quantum joint typicality conjecture~\cite{Sen18}.\\

{\it Acknowledgements.} We thank Kun Fang and Xi Wang for discussions and Joseph M.~Renes for pointing out a gap in the proof of a previous version of Theorem~\ref{thm:pa-qsi}. MT acknowledges an Australian Research Council Discovery Early Career Researcher Awards, project DE160100821. Part of the work done when R.J. was visiting Tata Institute of Fundamental Research (TIFR), Mumbai, India as a ``VAJRA Adjunt Faculty'' under the ``Visiting Advanced Joint Research (VAJRA) Faculty Scheme'' of the Science and Engineering Research Board (SERB), Department of Science and Technology (DST), Government of India. This work is supported by the Singapore Ministry of Education and the National Research Foundation, also through the Tier 3 Grant “Random numbers from quantum processes” MOE2012-T3-1-009.


\appendix

\section{More properties of partially smoothed information measures}

\begin{lemma}\label{lem:classical-function}
Let $\rho_{XB}=\sum_x|x\rangle\langle x|_X\otimes\rho_B^x\in\cS_\circ(XB)$, $\eps\in[0,1]$, and $f:X\to Z$ be a function. Then, we have for $\omega_{ZB}:=\sum_x|f(x)\rangle\langle f(x)|_Z\otimes\rho_B^x\in\cS_\circ(ZB)$ that
\begin{align}
H_{\min}^{\eps,P}(Z|\dot{B})_\omega\leq H_{\min}^{\eps,P}(X|\dot{B})_\rho\,.
\end{align}
Moreover, when $B = Y$ is classical then we also have $H_{\min}^{\eps,T}(Z|\dot{Y})_\omega\leq H_{\min}^{\eps,T}(X|\dot{Y})_\rho$.
\end{lemma}

\begin{proof}
  For purified distance the proof follows along similar lines as~\cite[Prop.~6.4]{mybook}. We first use of the invariance of the smooth conditional min-entropy under isometries (Lem.~\ref{lm:iso}) to assert that
  \begin{align}
    H_{\min}^{\eps,P}(X|\dot{B})_\rho = H_{\min}^{\eps,P}(XZ|\dot{B})_\omega \,, \label{eq:smoothed-iso}
  \end{align}
  where $\omega_{XZB}:=\sum_x |x\rangle\!\langle x|_X \otimes |f(x)\rangle\!\langle f(x)|_Z\otimes\rho_B^x$. Moreover, by~\cite[Lem.~5.3]{mybook}, we have
  \begin{align}
  H_{\min}(XZ|B)_\omega \geq H_{\min}(Z|B)_\omega \,. \label{eq:unsmoothed}
  \end{align}
  To lift this to smooth entropies, let us assume that $\tilde{\omega}_{ZB}$ achieves the maximum in the definition of $H_{\min}^{\eps,P}(Z|\dot{B})_\omega$. Then, by Uhlmann's theorem (specifically by~\cite[Cor.~3.1]{mybook}),  there exists a state $\tilde{\omega}_{XZB}$ that extends $\tilde{\omega}_{ZB}$ and is $\eps$-close to $\omega_{XZB}$. Moreover, by the monotonicity of the purified distance under measurements we can ensure that $\omega_{XZB}$ is classical on $Z$. Finally, Eq.~\eqref{eq:unsmoothed} applied to the state $\omega_{XZB}$ then yields
  \begin{align}
  H_{\min}^{\eps, P}(XZ|\dot{B})_{\omega} \geq H_{\min}(XZ|B)_{\tilde{\omega}} \geq H_{\min}(Z|B)_{\tilde{\omega}} = H_{\min}^{\eps,P}(Z|\dot{B})_{\omega} \,, \label{eq:concludehere}
  \end{align}
   concluding the proof of the first statement.
    
  The proof for the classical case and generalized trace distance is adapted from~\cite[Lem.~2]{watanabe12}. First, note that \eqref{eq:smoothed-iso} holds for generalized trace distance. Given~\eqref{eq:unsmoothed}, it thus remains to extend the optimizer $\tilde{\omega}_{ZY}$ to a suitable $\tilde{\omega}_{XZY}$ without invoking Uhlmann's theorem. As in~\cite{watanabe12}, we define
  \begin{align}
    \tilde{\omega}_{XZY}(x,z,y) := \frac{\omega_{XZY}(x,z,y)}{\omega_{ZY}(z,y)} \, \tilde{\omega}_{ZY}(z,y) = \omega_{X|ZY}(x|z,y) \, \tilde{\omega}_{ZY}(z,y) \,.
  \end{align}
  Using the definition of the generalized trace distance we then find that
  \begin{align}
    T(\tilde{\omega}_{XZY}, \omega_{XZY}) &= \sum_{x,y,z:\atop \omega_{XZY}(x,z,y) \geq \tilde{\omega}_{XZY}(x,z,y)} \omega_{XZY}(x,z,y) - \tilde{\omega}_{XZY}(x,z,y)  \\
    &= \sum_{x,y,z: \atop \omega_{ZY}(z,y) \geq \tilde{\omega}_{ZY}(z,y)} \omega_{X|ZY}(x|z,y) \Big( \omega_{ZY}(z,y) - \tilde{\omega}_{ZY}(z,y) \Big) \\
    &= \sum_{y,z: \atop \omega_{ZY}(z,y) \geq \tilde{\omega}_{ZY}(z,y)} \omega_{ZY}(z,y) - \tilde{\omega}_{ZY}(z,y) \\
    &= T(\tilde{\omega}_{ZY}, \omega_{ZY}) \,.
  \end{align}
  The proof concludes with the same argument given in~\eqref{eq:concludehere}.
\end{proof}

We would also like to note here that for generalized trace distance the above lemma does in fact not hold in the quantum case (when asking for identical smoothing parameters for both min-entropies) and a counter-example can be constructed using the states given in~\cite[p.~5]{renes17}. 

\begin{lemma}\label{lem:tripartite-dimension}
Let $\rho_{ABR}\in\cS_{\circ}(ABR)$ and $\eps\in[0,1]$. Then, we have
\begin{align}
H_{\min}^{\eps,P}(AB|\dot{R})_\rho\leq H_{\min}^{\eps,P}(A|\dot{R})_\rho+\log|B|\,.
\end{align}
\end{lemma}

\begin{proof}
This is implied by~\cite[Lem.~3.9]{berta08}.
\end{proof}

\begin{lemma}\label{lem:coherent}
Let $\rho_{ABXX'}\in\cS_{\circ}(ABXX')$ be coherently classical on $XX'$ and $\eps\in[0,1]$. Then, we have
\begin{align}
I_{\max}^{\eps,P}(BXX';\dot{A})_{\rho}\leq I_{\max}^{\eps,P}(BX';\dot{A})_\rho+\log|X|\,.
\end{align}
\end{lemma}

\begin{proof}
Let $\tilde{\rho}_{ABX'}\in\cS_{\bullet}(ABX')$ and $\sigma_{BX'}\in\cS_{\circ}(BX')$ be the optimizers in $I_{\max}^{\eps,P}(BX';\dot{A})_\rho$, where both can assumed to be classical on $X'$. Now taking an extension $\tilde{\rho}_{ABXX'}$ of $\tilde{\rho}_{ABX'}$ with $P(\tilde{\rho}_{ABXX'},\rho_{ABXX'})\leq\eps$ we get
\begin{align}\label{eq:apply}
\tilde{\rho}_{ABXX'}\leq|X|\cdot\id_{X}\otimes\tilde{\rho}_{ABX'}\leq|X|\cdot2^{D_{\max}(\tilde{\rho}_{ABX'}\|\rho_{A}\otimes\sigma_{BX'})}\cdot\id_{X}\otimes\rho_{A}\otimes\sigma_{BX'}\,.
\end{align}
Applying $\Pi_{XX'}:=\sum_{X}|x\rangle\langle x|_{X}\otimes|x\rangle\langle x|_{X'}$ to both sides of Eq.~\eqref{eq:apply} leads to
\begin{align}\label{eq:apply2}
\Pi_{XX'}\tilde{\rho}_{ABXX'}\Pi_{XX'}\leq|X|\cdot2^{D_{\max}(\tilde{\rho}_{ABX'}\|\rho_{A}\otimes\sigma_{BX'})}\cdot\rho_{A}\otimes\Big(\Pi_{XX'}(\id_{X}\otimes\sigma_{BX'})\Pi_{XX'}\Big)\,.
\end{align}
For $\hat{\rho}_{ABXX'}:=\Pi_{XX'}\tilde{\rho}_{ABXX'}\Pi_{XX'}$ and $\hat{\sigma}_{BXX'}:=\Pi_{XX'}(\id_{X}\otimes\sigma_{BX'})\Pi_{XX'}$ this is 
\begin{align}
\hat{\rho}_{ABXX'}\leq|X|\cdot2^{D_{\max}(\tilde{\rho}_{ABX'}\|\rho_A\otimes\sigma_{BX'})}\cdot\rho_{A}\otimes\hat{\sigma}_{BXX'}\,,
\end{align}
which in turn implies
\begin{align}
2^{D_{\max}(\hat{\rho}_{ABXX'}\|\rho_A\otimes\hat{\sigma}_{BXX'})}\leq|X|\cdot2^{D_{\max}(\tilde{\rho}_{ABX'}\|\rho_A\otimes\sigma_{BX'})}\,.
\end{align}
Since $\Pi_{XX'}$ is trace non-increasing we have $\hat{\rho}_{ABXX'}\in\cS_{\bullet}(ABXX')$ with $P(\hat{\rho}_{ABXX'},\rho_{ABXX'})\leq\eps$ and together with $\hat{\sigma}_{BXX'}\in\cS_{\bullet}(BXX')$ this finishes the proof.
\end{proof}

\begin{lemma}\label{lem:projective-measurement}
Let $\rho_{AB}\in\cS_\circ(AB)$, $\sigma_B\in\cS_{\circ}(B)$, $\{ P_A^x \}$ be a projective measurement on $A$, and $\omega_{ABX}:=\sum_x |x\rangle\langle x|_X\otimes P_A^x\rho_{AB}P_A^x$. Then, we find that
\begin{align}\label{eq:lemma-first}
D_{\max}(\rho_{AB}\|1_A\otimes\sigma_B)\leq D_{\max}(\omega_{ABX}\|1_A\otimes\sigma_B\otimes\omega_X)\,.
\end{align}
Moreover, for $\eps\in[0,1]$ we have
\begin{align}\label{eq:lemma-second}
H_{\min}^{\eps,P}(A|\dot{B})_\rho\geq\sup_{\tilde{\omega}_{ABX}}-D_{\max}(\tilde{\omega}_{ABX}\|1_A\otimes\rho_B\otimes\tilde{\omega}_X)\,,
\end{align}
where the supremum is over all classical-quantum $\tilde{\omega}_{ABX}\in\cS_{\circ}(ABX)$ with $P(\tilde{\omega}_{ABX},\omega_{ABX})\leq\eps$.
\end{lemma}

\begin{proof}
Eq.~\eqref{eq:lemma-first} is~\cite[Prop.~4.9]{berta08}. For Eq.~\eqref{eq:lemma-second} note that the isometric purification of $\{P_{A}^{x}\}$, $V=\sum_{x}|x\rangle_{X}\otimes P_{A}^{x}$, can be inverted on the image of $V$
\begin{align}
P_{V}=\sum_{x}P_{A}^{x}\otimes|x\rangle\langle x|_X\,.
\end{align}
Now, for every classical-quantum $\tilde{\omega}_{ABX}\in\cS_{\bullet}(ABX)$ with $P(\tilde{\omega}_{ABX},\omega_{ABX})\leq\eps$ we have for $\tilde{\omega}^V_{ARX}:=P_{V}\tilde{\omega}_{ABX}P_V$ that
\begin{align}
D_{\max}(\tilde{\omega}^V_{ABX}\|1_A\otimes\rho_B\otimes\tilde{\omega}_{X}^V)\leq D_{\max}(\tilde{\omega}_{ARX}\|1_A\otimes\rho_B\otimes\tilde{\omega}_{X})\,.
\end{align}
Hence, we can restrict the supremum in Eq.~\eqref{eq:lemma-second} to states in the image of $V$ and the claim now follows by Eq.~\eqref{eq:lemma-first} together with the invariance under isometries.
\end{proof}


\section{Assorted additional lemmas}\label{app:lemmas}

\begin{lemma}[Special case of Lem.~A.1 in~\cite{mythesis}]\label{lm:someineq}
Let $\rho_{AB} \in \cS_{\bullet}(AB)$ and $L: B \to B'$ be a contraction. Then, we have
\begin{align}
\tr_{B'} \left[ ( \id_A \otimes L ) \rho_{AB} ( \id_A \otimes L )^{\dag} \right] \leq \rho_A\,.
\end{align}
\end{lemma}

\begin{lemma}[Def.~3.3 \& 3.4 in~\cite{mybook}]\label{lem:tracedistprob}
Let $P_X, Q_X \in \cS_{\bullet}(X)$. Then, for $\mathcal{X}$ the set associated to $X$ we have
\begin{align}
T(P_X, Q_X) = \max_{S\in \mathcal{X}}|P_X(S) - Q_X(S)|\,.
\end{align}
\end{lemma}


\begin{lemma}[Variation of convex-split lemma from~\cite{Anurag17}]\label{convsplitlemm}
Let $\eps, \delta\in (0,1)$ and $\rho_{AB}, \rho'_{AB}\in \cS_{\circ}(AB),  \sigma_B\in \cS_{\circ}(B)$ such that $\Delta(\rho_{AB},\rho'_{AB})\leq \eps$. Then, for the quantum state
\begin{align}
\tau_{AB_1\ldots B_{2^R}} = \frac{1}{2^R}\sum_{j=1}^{2^R}\rho_{AB_j}\otimes \sigma_{B_1}\otimes \ldots \sigma_{B_{j-1}}\otimes\sigma_{B_{j+1}}\otimes\ldots \sigma_{B_{2^R}}&\\
\text{with $R\geq\left\lceil D_{\max}(\rho'_{AB}\|\rho'_A\otimes \sigma_B)+2\log\frac{2}{\delta}\right\rceil$}&
\end{align}
we have $\Delta(\tau_{AB_1\ldots B_{2^R}},\rho'_A\otimes\sigma_{B_1}\otimes\ldots\sigma_{B_{2^R}})\leq\eps+\delta$.
\end{lemma}

\begin{proof}
We only sketch the minor additional steps compared to the proof in~\cite{Anurag17}. For
\begin{align}
\tau'_{AB_1\ldots B_{2^R}} = \frac{1}{2^R}\sum_{j=1}^{2^R}\rho'_{AB_j}\otimes \sigma_{B_1}\otimes \ldots \sigma_{B_{j-1}}\otimes\sigma_{B_{j+1}}\otimes\ldots \sigma_{B_{2^R}}
\end{align}
we have from~\cite{Anurag17} that 
$$P(\tau'_{AB_1\ldots B_{2^R}},  \rho'_A\otimes \sigma_{B_1} \otimes\ldots \sigma_{B_{2^R}})\leq \delta\,,$$
which implies by the Fuchs-Van de Graaf inequality that
$$T(\tau'_{AB_1\ldots B_{2^R}},  \rho'_A\otimes \sigma_{B_1} \otimes\ldots \sigma_{B_{2^R}})\leq \delta\,.$$
Now, by the concavity of fidelity
$$F(\tau_{AB_1\ldots B_{2^R}}, \tau'_{AB_1\ldots B_{2^R}})\geq F(\rho_{AB}, \rho'_{AB})\quad\implies\quad P(\tau_{AB_1\ldots B_{2^R}}, \tau'_{AB_1\ldots B_{2^R}})\leq P(\rho_{AB}, \rho'_{AB})\,,$$
and by the triangle inequality $T(\tau_{AB_1\ldots B_{2^R}}, \tau'_{AB_1\ldots B_{2^R}})\leq T(\rho_{AB}, \rho'_{AB})$. The proof now follows by the triangle inequality for either $P$ or $T$.
\end{proof}

\bibliographystyle{ultimate}
\bibliography{library}

\end{document}